\documentclass[11pt]{article}
\pdfoutput=1

\usepackage[vmargin=30mm,hmargin=33mm]{geometry}
\usepackage{amsmath,amssymb,amsthm,booktabs,color,doi,enumitem,graphicx,latexsym,url,xcolor,thmtools,thm-restate}
\usepackage[numbers,sort&compress]{natbib}
\usepackage{microtype,hyperref}
\usepackage{comment}

\definecolor{citecolor}{HTML}{0000C0}
\definecolor{urlcolor}{HTML}{000080}

\hypersetup{
    colorlinks=true,
    linkcolor=black,
    citecolor=citecolor,
    filecolor=black,
    urlcolor=urlcolor,
}

\declaretheorem{theorem}
\declaretheorem[numberlike=theorem]{lemma}
\declaretheorem[numberlike=theorem]{corollary}

\newcommand{\namedref}[2]{\hyperref[#2]{#1~\ref*{#2}}}
\newcommand{\sectionref}[1]{\namedref{Section}{#1}}
\newcommand{\theoremref}[1]{\namedref{Theorem}{#1}}
\newcommand{\corollaryref}[1]{\namedref{Corollary}{#1}}
\newcommand{\lemmaref}[1]{\namedref{Lemma}{#1}}

\newcommand{\Z}{\mathbb{Z}}
\newcommand{\card}[1]{\left\lvert {#1} \right\rvert}

\newcommand{\mme}{\rho}
\newcommand{\mmeaux}{\sigma} 
\DeclareMathOperator{\trace}{tr}

\setitemize{noitemsep}
\setenumerate{noitemsep,itemsep=1ex}

\newcommand*\samethanks[1][\value{footnote}]{\footnotemark[#1]}

\begin{document}
\hypersetup{
    pdfauthor={Keren Censor-Hillel, Petteri Kaski, Janne H.\ Korhonen, Christoph Lenzen, Ami Paz, Jukka Suomela},
    pdftitle={Algebraic Methods in the Congested Clique},
}

\begin{titlepage}
\setcounter{page}{0}
\title{Algebraic Methods in the Congested Clique}

\author{Keren Censor-Hillel\thanks{Department of Computer Science, Technion, \textsf{\{ckeren, amipaz\}@cs.technion.ac.il}.}
\and Petteri Kaski\thanks{Helsinki Institute for Information Technology HIIT \& Department of Information and Computer Science, Aalto University, \textsf{\{petteri.kaski, jukka.suomela\}@aalto.fi}.}
\and Janne H.\ Korhonen\thanks{Helsinki Institute for Information Technology HIIT \& Department of Computer Science, University of Helsinki, \textsf{janne.h.korhonen@helsinki.fi}.}
\and Christoph Lenzen\thanks{Department for Algorithms and Complexity, MPI Saarbr\"{u}cken, \textsf{clenzen@mpi-inf.mpg.de}.}
\and Ami Paz\samethanks[1]
\and Jukka Suomela\samethanks[2]
}
\date{}
\maketitle

\begin{abstract}
In this work, we use algebraic methods for studying distance computation and subgraph detection tasks in the \emph{congested clique} model. Specifically, we adapt parallel matrix multiplication implementations to the congested clique, obtaining an $O(n^{1-2/\omega})$ round matrix multiplication algorithm, where $\omega < 2.3728639$ is the exponent of matrix multiplication. In conjunction with known techniques from centralised algorithmics, this gives significant improvements over previous best upper bounds in the congested clique model.
The highlight results include:
\begin{itemize}
    \item[--] triangle and 4-cycle counting in $O(n^{0.158})$ rounds, improving upon the $O(n^{1/3})$ algorithm of Dolev et al. [DISC 2012],
    \item[--] a $(1 + o(1))$-approximation of all-pairs shortest paths in $O(n^{0.158})$ rounds, improving upon the $\tilde{O} (n^{1/2})$-round $(2 + o(1))$-approximation algorithm of Nanongkai [STOC 2014], and
    \item[--] computing the girth in $O(n^{0.158})$ rounds, which is the first non-trivial solution in this model.
\end{itemize}
In addition, we present a novel constant-round combinatorial algorithm for detecting 4-cycles.

\end{abstract}

\thispagestyle{empty}
\end{titlepage}
\setcounter{page}{1}



\begin{table}[b!]
\newcommand{\myitem}{\ \ $\cdot$ }
\centering
\begin{tabular*}{\linewidth}{@{}l@{\extracolsep{\fill}}l@{}l@{}r@{}}
\toprule
& \multicolumn{3}{@{}c}{Running time} \\
\cmidrule{2-4}
Problem & This work & Prior work & \\
\midrule
matrix multiplication (semiring) & $O(n^{1/3})$ & --- & \\
matrix multiplication (ring) & $O(n^{0.158})$ & $O(n^{0.373})$ & \cite{drucker13} \\
\midrule
triangle counting & $O(n^{0.158})$ & $O(n^{1/3}/\log n)$ & \cite{tritri} \\
4-cycle detection & $O(1)$ & $O(n^{1/2}/\log n)$ & \cite{tritri} \\
4-cycle counting & $O(n^{0.158})$ & $O(n^{1/2}/\log n)$ & \cite{tritri} \\
$k$-cycle detection & $2^{O(k)} n^{0.158}$ & $O(n^{1-2/k} / \log n)$ & \cite{tritri} \\
girth & $O(n^{0.158})$ & --- & \\
\midrule
weighted, directed APSP & $O(n^{1/3} \log n)$ & --- & \\
\myitem weighted diameter $U$ & $O(Un^{0.158})$ & --- & \\
\myitem $(1+o(1))$-approximation & $O(n^{0.158})$ & --- & \\
\myitem $(2+o(1))$-approximation & & $\tilde{O}(n^{1/2})$ & \cite{nanongkai14} \\
\midrule
unweighted, undirected APSP & $O(n^{0.158})$ & --- & \\
\myitem $(2+o(1))$-approximation & & $\tilde{O}(n^{1/2})$ & \cite{nanongkai14} \\
\bottomrule
\end{tabular*}
\caption{Our results versus prior work, for the currently best known bound $\omega < 2.3729$~\cite{legall2014powers}; $\tilde O$ notation hides polylogarithmic factors.}\label{tab:results}
\end{table}

\section{Introduction}\label{sec:intro}
Algebraic methods have become a recurrent tool in centralised algorithmics, employing a wide range of techniques (e.g.,\ \cite{Bjorklund14,
BjorklundH14,
BjorklundHKK07,
BjorklundHK09,
BjorklundKK13,
BjorklundKK14,
BodlaenderCKN13,
CyganKN13,
CyganNPPRW11,
CzumajL07,
EisenbrandG04,
FominLRSR12,
FominLS14,
Koutis08,
KowalukLL11,
LokshtanovN10,
nevsetvril1985complexity,
Williams09,
vassilevska2013finding}). %
In this paper, we bring techniques from the algebraic toolbox to the aid of distributed computing, by leveraging fast matrix multiplication in the \emph{congested clique} model.

In the congested clique model, the $n$ nodes of a graph $G$ communicate by exchanging messages of $O(\log{n})$ size in a \emph{fully-connected} synchronous network; initially, each node is aware of its neighbours in $G$. In comparison with the traditional CONGEST model \cite{peleg00}, the key difference is that a pair of nodes can communicate directly even if they are not adjacent in graph~$G$. The congested clique model masks away the effect of \emph{distances} on the computation and focuses on the limited \emph{bandwidth}. As such, it has been recently gaining increasing attention~\cite{PemmarajuS14_MST_logloglogn,tritri,drucker13,lenzen2013optimal,lotker05,nanongkai14, LenzenW11,patt-shamir11,hegeman14,HegemanP14}, in an attempt to understand the relative computational power of distributed computing models.

The key insight of this paper is that matrix multiplication algorithms from parallel computing can be adapted to obtain an $O(n^{1-2/\omega})$ round matrix multiplication algorithm in the congested clique, where $\omega < 2.3728639$ is the matrix multiplication exponent~\cite{legall2014powers}. Combining this with well-known centralised techniques allows us to use fast matrix multiplication to solve various combinatorial problems, immediately giving $O(n^{0.158})$-time algorithms in the congested clique for many classical graph problems. Indeed, while most of the techniques we use in this work are known beforehand, their combination gives significant improvements over the best previously known upper bounds. Table~\ref{tab:results} contains a summary of our results, which we overview in more details in what follows.

\subsection{Matrix Multiplication on a Congested Clique}

As a basic primitive, we consider the computation of the product $P = ST$ of two $n \times n$ matrices $S$ and $T$ on a congested clique of $n$ nodes. We will tacitly assume that the matrices are initially distributed so that node $v$ has row $v$ of both $S$ and $T$, and each node will receive row $v$ of $P$ in the end. Recall that the matrix multiplication exponent $\omega$ is defined as the infimum over $\mmeaux$ such that product of two $n \times n$ matrices can be computed with $O(n^{\mmeaux})$ arithmetic operations; it is known that $2 \le \omega < 2.3728639$~\cite{legall2014powers}, and it is conjectured, though not unanimously, that $\omega = 2$.

\begin{restatable}{theorem}{thmmm}\label{thm:mm}
The product of two matrices $n \times n$ can be computed in a congested clique of $n$ nodes in
$O(n^{1/3})$ rounds over semirings. Over rings, this product can be computed in $O(n^{1-2/\omega+\varepsilon})$ rounds for any constant $\varepsilon>0$.
\end{restatable}

Theorem~\ref{thm:mm} follows by adapting known parallel matrix multiplication algorithms for semirings~\cite{AgarwalBGJP95_3d,mccoll1995scalable} and rings~\cite{LuoD95_layout_parallel,mccoll1996,tiskin-phd,BallardDHLS12_Strassen_upper} to the clique model, via the routing technique of~\citet{lenzen2013optimal}. In fact, with little extra work one can show that the resulting algorithm is also \emph{oblivious}, that is, the communication pattern is predefined and does not depend on the input matrices.
Hence, the oblivious routing technique of~\citet{tritri} suffice for implementing these matrix multiplication algorithms.

The above addresses matrices whose entries can be encoded with $O(\log n)$ bits, which is sufficient for dealing with integers of absolute value at most $n^{O(1)}$. In general, if $b$ bits are sufficient to encode matrix entries, the bounds above hold with a multiplicative factor of $b / \log n$; for example, working with integers with absolute value at most $2^{n^{\varepsilon}}$ merely incurs a factor $n^{\varepsilon}$ overhead in running times.

\paragraph{Distributed matrix multiplication exponent.}
Analogously with the matrix multiplication exponent, we denote by $\mme$ the exponent of matrix multiplication in the congested clique model, that is, the infimum over all values $\mmeaux$ such that there exists a matrix multiplication algorithm in the congested clique running in $O(n^{\mmeaux})$ rounds. In this notation, \theoremref{thm:mm} gives us
\[ \mme \le 1 - 2/\omega < 0.15715\,;\]
prior to this work, it was known that $\mme \leq \omega-2$~\cite{drucker13}.

For the rest of this paper, we will -- analogously with the convention in centralised algorithmics -- slightly abuse this notation by writing $n^\mme$ for the complexity of matrix multiplication in the congested clique. This hides factors up to $O(n^\varepsilon)$ resulting from the fact that the exponent $\mme$ is defined as infimum of an infinite set.

\paragraph{Lower bounds for matrix multiplication.}
The matrix multiplication results are optimal in the sense that for any sequential matrix multiplication implementation, any scheme for simulating that implementation in the congested clique cannot give a faster algorithm than the construction underlying \theoremref{thm:mm}; this follows from known results for parallel matrix multiplication~\cite{BallardDHS12_strassen_lower, IronyTT04_3d_lower,AggarwalCS90_PRAM,tiskin1998}. Moreover, we note that for the \emph{broadcast congested clique model}, where each node is required to send the same message to all nodes in any given round, recent lower bounds \cite{arXiv:1412.3445} imply that matrix multiplication cannot be done faster than $\tilde\Omega(n)$ rounds. 

\subsection{Applications in Subgraph Detection}

\paragraph{Cycle detection and counting.}
Our first application of fast matrix multiplication is to the problems of triangle counting~\cite{itai1978finding} and 4-cycle counting. 

\begin{restatable}{corollary}{thmtriangles}\label{cor:triangles}
    For directed and undirected graphs, the number of triangles and 4-cycles can be computed in $O(n^{\mme})$ rounds.
\end{restatable}

For $\mme \le 1-2/\omega$, this is an improvement upon the previously best known $O(n^{1/3})$-round triangle detection algorithm of \citet{tritri} and an $O(n^{\omega - 2 + \varepsilon})$-round algorithm of \citet{drucker13}. Indeed, we disprove the conjecture of~\citet{tritri} that any deterministic oblivious algorithm for detecting triangles requires $\tilde{\Omega}(n^{1/3})$ rounds.

When only detection of cycles is required, we observe that combining the fast distributed matrix multiplication with the well-known technique of \emph{colour-coding} \cite{alon1995color} allows us to detect $k$-cycles in $\tilde{O}(n^{\mme})$ rounds for any constant $k$. This improves upon the subgraph detection algorithm of Dolev et al.~\cite{tritri}, which requires $\tilde{O}(n^{1-2/k})$ rounds for detecting subgraphs of $k$ nodes. However, we do not improve upon the algorithm of Dolev et al. for general subgraph detection.

\begin{restatable}{theorem}{thmkcycles}\label{thm:k-cycles}
For directed and undirected graphs, the existence of $k$-cycles can be detected in $2^{O(k)} n^{\mme} \log n$ rounds.
\end{restatable}

For the specific case of $k=4$, we provide a novel algorithm that does not use matrix multiplication and detects 4-cycles in only $O(1)$ rounds.

\begin{restatable}{theorem}{thmfourcycles}\label{thm:4-cycles}
The existence of 4-cycles can be detected in $O(1)$ rounds.
\end{restatable}


\paragraph{Girth.}
We compute the girth of a graph by leveraging a known trade-off between the girth and the number of edges of the graph~\cite{Matousek02_geometry}. Roughly, we detect short cycles fast, and if they do not exist then the graph must have sufficiently few edges to be learned by all nodes. As far as we are aware, this is the first algorithm to compute the girth in this setting.

\begin{restatable}{theorem}{thmgirth}
For undirected, unweighted graphs, the girth can be computed in $\tilde{O}(n^{\mme})$ rounds.
\end{restatable}

\subsection{Applications in Distance Computation}

\paragraph{Shortest paths.}
The \emph{all-pairs shortest paths} problem (APSP) likewise admits algorithms based on matrix multiplication. The basic idea is to compute the $n^\text{th}$ power of the input graph's weight matrix over the min-plus semiring, by iteratively computing squares of the matrix~\cite{furman1970application, Munro197156, fm1971boolean}. 

\begin{restatable}{corollary}{corrapspsemi}\label{thm:apsp-semiring}
For weighted, directed graphs with integer weights in $\{0 ,\pm 1, \dotsc, \pm M \}$, all-pairs shortest paths can be computed in $O(n^{1/3} \log n \lceil\log M/\log n\rceil)$ communication rounds.
\end{restatable}

We can leverage fast ring matrix multiplication to improve upon the above result; however, the use of ring matrix multiplication necessitates some trade-offs or extra assumptions. For example, for unweighted and undirected graphs, it is possible to recover the exact shortest paths from powers of the adjacency matrix over the Boolean semiring~\cite{Seidel1995400}.

\begin{restatable}{corollary}{corseidel}
For undirected, unweighted graphs, all-pairs shortest paths can be computed in $\tilde{O}(n^{\mme})$ rounds.
\end{restatable}

For small integer weights, we use the well-known idea of embedding a min-plus semiring matrix product into a matrix product over a ring; this gives a multiplicative factor to the running time proportional to the length of the longest path.

\begin{restatable}{corollary}{cordiamapsp}
For directed graphs with positive integer weights and weighted diameter $U$, all-pairs shortest paths can be computed in $\tilde{O}(U n^{\mme})$ rounds.
\end{restatable}

While this corollary is only relevant for graphs of small weighted diameter, the same idea can be combined with weight rounding~\cite{raghavan85,zwick2002all,nanongkai14} to obtain a fast approximate APSP algorithm without such limitations.

\begin{restatable}{theorem}{thmapxapsp}
For directed graphs with integer weights in $\{0,1,\ldots,2^{n^{o(1)}}\}$, we can compute $(1 + o(1))$-approximate all-pairs shortest paths in $O(n^{\mme + o(1)})$ rounds.
\end{restatable}

For comparison, the previously best known combinatorial algorithm for APSP on the congested clique achieves a $(2+o(1))$-approximation in $\tilde{O}(n^{1/2})$ rounds~\cite{nanongkai14}.

\subsection{Additional Related Work}\label{sec:related-work}
Computing distances in graphs,
such as the diameter, all-pairs shortest paths (APSP), and single-source shortest paths (SSSP)
are fundamental problems in most computing settings.
The reason for this lies in the abundance of applications of such computations, evident also by the huge amount of research dedicated to it~\cite{Chan10_apsp,HanT12_apsp,
Han08_apsp,Takaoka04_apsp,Zwick06_APSP,Takaoka05_apsp,Chan08_apsp,Fredman76_apsp,
williams2014apsp,
Zwick01_graph_distances_survey,zwick2002all}.

In particular, computing graph distances is vital for many distributed applications and, as such, has been widely studied in the CONGEST model of computation~\cite{peleg00}, where $n$ processors located in $n$ distinct nodes of a graph $G$ communicate over the graph edges using $O(\log{n})$-bit messages. Specifically, many algorithms and lower bounds were given for computing and approximating graph distances
in this setting~\cite{DHKNPPW-11, nanongkai14, LP13:podc,HolzerW12, PelegRT12,FHW-12,LenzenP13_routing_tables,holzer14,KP98,PelegR-00}. Some lower bounds apply even for graphs of small diameter; 
however, these lower bound constructions boil down to graphs that contain \emph{bottleneck} edges limiting the amount of information that can be exchanged between different parts of the graph quickly.

The intuition that the congested clique model would abstract away distances and bottlenecks and bring to light only the congestion challenge has proven inaccurate.
Indeed, a number of tasks have been shown to admit sub-logarithmic or even constant-round solutions, exceeding by far what is possible in the CONGEST model with only low diameter. 
The pioneering work of \citet{lotker05} shows that a minimum spanning tree (MST) can be computed in $O(\log \log n)$ rounds. \citet{hegeman14} show how to construct a $3$-ruling set, with applications to maximal independent set and an approximation of the MST in certain families of graphs; sorting and routing have been recently addressed by various authors~\cite{lenzen2013optimal,LenzenW11,patt-shamir11}. A connection between the congested clique model and the MapReduce model is discussed by~\citet{HegemanP14}, where algorithms are given for colouring problems. On top of these positive results, Drucker et al.~\cite{drucker13} recently proved that essentially \emph{any} non-trivial unconditional lower bound on the congested clique would imply novel circuit complexity lower bounds.

The same work also points out the connection between fast matrix multiplication algorithms and triangle detection in the congested clique. Their construction yields an $O(n^{\omega -2+\varepsilon})$ round algorithm for matrix multiplication over rings in the congested clique model, giving also the same running bound for triangle detection; if $\omega = 2$, this gives $\mme = 0$, matching our result. However, with the currently best known centralised matrix multiplication algorithm, the running time of the resulting triangle detection algorithm is $O(n^{0.3729})$ rounds, still slower than the combinatorial triangle detection of \citet{tritri}, and if $\omega>2$, the solution presented in this paper is faster.



\section{Matrix Multiplication Algorithms}

In this section, we consider computing the product $P = ST$ of two $n \times n$ matrices $S = (S_{ij})$ and $T = (T_{ij})$ on the congested clique with $n$ nodes. For convenience, we tacitly assume that nodes $v \in V$ are identified with $\{ 1, 2, \dotsc, n\}$, and use nodes $v \in V$ to directly index the matrices. The local input in the matrix multiplication task for each node $v \in V$ is the row $v$ of both $S$ and $T$, and the at the end of the computation each node $v \in V$ will output the row $v$ of $P$. However, we note that the exact distribution of the input and output is not important, as we can re-arrange the entries in constant rounds as long as each node has $O(n)$ entries~\cite{lenzen2013optimal}.

\thmmm*

\theoremref{thm:mm} follows directly by simulating known parallel matrix multiplication algorithms in the congested clique model using a result of \cite{lenzen2013optimal}. This work discusses simulation of the \emph{bulk-synchronous parallel} (BSP) model, which we can use to obtain \theoremref{thm:mm} as a corollary from known BSP matrix multiplication results~\cite{mccoll1995scalable,mccoll1996,tiskin-phd}. However, essentially the same matrix multiplication algorithms have been widely studied in various parallel computation models, and the routing scheme underlying the simulation result of \cite{lenzen2013optimal} allows also simulation of these other models on the congested clique:
\begin{itemize}
    \item[--] The first part of \theoremref{thm:mm} is based on the so-called parallel 3D matrix multiplication algorithm~\cite{AgarwalBGJP95_3d,mccoll1995scalable}, essentially a parallel implementation of the school-book matrix multiplication; alternatively, the same algorithm can be obtained by slightly modifying the triangle counting algorithm of \citet{tritri}.
    \item[--] The second part uses a scheme that allows one to adapt any bilinear matrix multiplication algorithm into a fast parallel matrix multiplication algorithm~\cite{LuoD95_layout_parallel,mccoll1996,tiskin-phd,BallardDHLS12_Strassen_upper}.
\end{itemize}

A more detailed examination in fact shows that the matrix multiplication algorithms are \emph{oblivious}, that is, the communication pattern is pre-defined and only the content of the messages depends on the input. This further allows us to use the static routing scheme of~\citet{tritri}, resulting in simpler algorithms with smaller constant factors in the running time.

To account for all the details, and to provide an easy access for readers not familiar with the parallel computing literature, we present the congested clique versions of these algorithms in full detail in Sections~\ref{sec:smm} and~\ref{sec:fmm}.

\subsection{Semiring matrix multiplication}\label{sec:smm}

\newcommand{\emptyind}{\mathord{\ast}}

\newcommand{\fmbdual}[2]{{#1}{#2}\emptyind}
\newcommand{\fmbsingle}[1]{{#1}\emptyind}

\newcommand{\mbI}[1]{{#1}\emptyind\emptyind}
\newcommand{\mbII}[1]{{#1}\emptyind\emptyind}
\newcommand{\mbIII}[1]{{#1}\emptyind\emptyind}

\newcommand{\pbI}[1]{{#1}\emptyind\emptyind}
\newcommand{\pbII}[1]{\emptyind{#1}\emptyind}
\newcommand{\pbIII}[1]{\emptyind\emptyind{#1}}

\paragraph{Preliminaries.}
For convenience, let us assume that the number of nodes is such that $n^{1/3}$ is an integer. We view each node $v \in V$  as a three-tuple $v_1 v_2 v_3$ where $v_1, v_2, v_3 \in [n^{1/3}]$; for concreteness, we may think that $v_1 v_2 v_3$ is the representation of $v$ as a three-digit number in base-$n^{1/3}$.

For a matrix $S$ and index sets $U,W \subseteq V$, use the notation $S[U,W]$ to refer to the submatrix obtained by taking all rows $u$ with $u \in U$ and columns $w$ with $w \in W$. To easily refer to specific subsets of indices, we use $\emptyind$ as a wild-card in this notation; specifically, we use notation $\pbI{x} = \{ v \colon v_1 = x \}$, $\pbII{x} = \{ v \colon v_2 = x \}$ and $\pbIII{x} = \{ v \colon v_3 = x \}$. Finally, in conjunction with this notation, we use the shorthand $*$ to denote the whole index set $V$ and $v$ to refer to a singleton set $\{ v \}$. See Figure~\ref{fig:semiring-mm}.

\begin{figure}
    \centering
    \includegraphics[page=2]{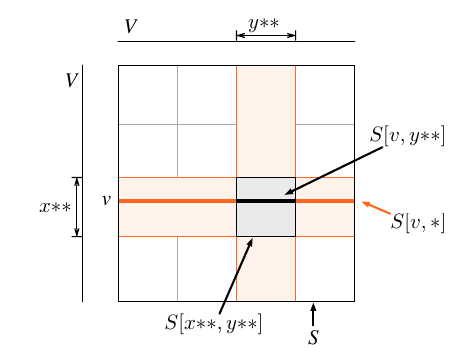}
    \caption{Semiring matrix multiplication: partitioning scheme for matrix entries.}\label{fig:semiring-mm}
\end{figure}

\paragraph{Overview.}
The distributed implementation of the school-book matrix multiplication we present is known as the 3D algorithm. To illustrate why, we note that the $n^3$ element-wise multiplications of the form
\[ P_{uw} = S_{uv}T_{vw}\,,\hspace{10mm} u,v,w \in V \]
can be viewed as points in the cube $V \times V \times V$. To split the element-wise multiplications equally among the nodes, we partition this cube into $n$ subcubes of size $n^{2/3} \times n^{2/3} \times n^{2/3}$. Specifically, each node $v$ is assigned the subcube $\mbI{v_1}\times \mbI{v_2} \times \mbI{v_3}$, corresponding to the multiplication task
\[ S[\mbI{v_1},\mbII{v_2}] T[\mbII{v_2},\mbIII{v_3}]\,.\]

\paragraph{Algorithm description.}
The algorithm computes $n \times n$ intermediate matrices $P^{(w)} = S[\emptyind, \mbII{w}]T[\mbII{w},\emptyind]$ for $w \in [n^{1/3}]$, so that each node $v$ computes the block
\[ P^{(v_2)}[\mbI{v_1},\mbI{v_3}] = S[\mbI{v_1}, \mbII{v_2}]T[\mbII{v_2},\mbIII{v_3}]\,.\]
Specifically, this is done as follows.
\begin{description}
\item[Step 1: Distributing the entries.] Each node $v\in V$ sends, for each node $u \in \pbI{v_1}$, the submatrix $S[v,\mbII{u_2}]$ to node $u$, and for each node $w \in \pbII{v_2}$, the submatrix $T[v,\mbIII{w_3}]$ to $w$.
Each such submatrix has size $n^{2/3}$ and there are $2 n^{2/3}$
recipients, for a total of $2 n^{4/3}$ messages per node.

Dually, each node $v\in V$ receives the submatrix
$S[\mbI{v_1},\mbII{v_2}]$ and the submatrix $T[\mbII{v_2},\mbIII{v_3}]$.
In particular, the submatrix $S[u,\mbII{v_2}]$ is
received from the node $u$ for $u \in \pbI{v_1}$, and the submatrix $T[w,\mbIII{v_3}]$
is received from the node $w \in \pbII{v_2}$. In total, each node receives $2 n^{4/3}$ messages.

\item[Step 2: Multiplication.] Each node $v \in V$ computes the product $S[\mbI{v_1},\mbII{v_2}]$ and $T[\mbII{v_2},\mbIII{v_3}]$ to get the $n^{2/3}\times n^{2/3}$ product matrix $P^{(v_2)}[\mbI{v_1},\mbIII{v_3}]$.

\item[Step 3: Distributing the products.]
Each node $v \in V$ sends submatrix
$P^{(v_2)}[u,\mbIII{v_3}]$ to each node $u \in \pbI{v_1}$.
Each such submatrix has size $n^{2/3}$ and there are $n^{2/3}$
recipients, for a total of $n^{4/3}$ messages per node.

Dually, each node $v\in V$ receives the submatrices
$P^{(w)}[v,\emptyind]$ for each $w\in[n^{1/3}]$. In particular,
the submatrix $P^{(u_2)}[v,\mbIII{u_3}]$ is received from the
node $u \in \pbI{v_1}$. The total number of received messages is $n^{4/3}$ per node.

\item[Step 4: Assembling the product.]
Each node $v\in V$ computes the submatrix
$P[v,\emptyind]=\sum_{w\in[n^{1/3}]} P^{(w)}[v,\emptyind]$
of the product $P=ST$.
\end{description}

\paragraph{Analysis.}
The maximal number of messages sent or received in one of the above steps is $O(n^{4/3})$. Moreover, the communication pattern clearly does not depend on the input matrices, so the algorithm can be implemented in oblivious way on the congested clique using the routing scheme of~\citet[Lemma 1]{tritri}; the running time is $O(n^{1/3})$ rounds.

\subsection{Fast Matrix Multiplication}\label{sec:fmm}

\paragraph{Bilinear matrix multiplication.}
Consider a \emph{bilinear algorithm} multiplying two $d \times d$
matrices using $m < d^3$ scalar multiplications,
such as the Strassen algorithm~\cite{strassen}.
Such an algorithm computes the matrix product $P = ST$
by first computing $m$ linear combinations of entries of both matrices,
\begin{equation}\label{eq:bilin-1}
\hat{S}^{(w)} = \sum_{(i,j) \in [d]^2}\alpha_{ijw} S_{ij}
\hspace{10mm}\text{and}\hspace{10mm}
\hat{T}^{(w)} = \sum_{(i,j) \in [d]^2}\beta_{ijw} T_{ij}
\end{equation}
for each $w \in [m]$,
then computing the products
$\hat{P}^{(w)} = \hat{S}^{(w)} \hat{T}^{(w)}$ for $w \in [m]$,
and finally obtaining $P$ as
\begin{equation}\label{eq:bilin-2}
P_{ij} = \sum_{w \in [m]}\lambda_{ijw} \hat{P}^{(w)}\,,
\hspace{10mm}\text{for $(i,j) \in [d]^2,$}
\end{equation}
where $\alpha_{ijw}$, $\beta_{ijw}$ and $\lambda_{ijw}$
are scalar constants that define the algorithm. In this section we show that any bilinear matrix multiplication algorithm can be efficiently translated to the congested clique model.

\begin{lemma}
Let $R$ be a ring, and assume there exists a family of bilinear matrix multiplication algorithms that can compute product of $n \times n$ matrices with $O(n^{\mmeaux})$ multiplications. Then matrix multiplication over $R$ can be computed in the congested clique in $O\bigl(n^{1 - 2/\mmeaux} (b / \log n) \bigr)$ rounds, where $b$ is the number of bits required for encoding a single element of $R$.
\end{lemma}

In particular for integers, rationals and their extensions, it is known that for any constant $\varepsilon > 0$ there is a bilinear algorithm for matrix multiplication that uses $O(n^{\omega+\varepsilon})$ multiplications~\cite{burgisser1997algebraic}; thus, the second part of \theoremref{thm:mm} follows from the above lemma.

\paragraph{Preliminaries.}
Let us fix a bilinear algorithm that computes the product of $d \times d$ matrices using $m(d) = O(d^\mmeaux)$ scalar multiplications for any $d$, where $2 \le \mmeaux \le 3$. To multiply two $n\times n$ matrices on a congested clique of $n$ nodes, fix $d$ so that $m(d) = n$, assuming for convenience that $n$ is such that this is possible. Note that we have $d = O(n^{1/\mmeaux})$.

Similarly with the semiring matrix multiplication, we view each node $v$ as three-tuple $v_1v_2v_3$, where we assume that $v_1 \in [d]$, $v_2 \in [n^{1/2}]$ and $v_3 \in [n^{1/2}/d]$; that is, $v_1 v_2 v_3$ can be viewed as a mixed-radix representation of the integer $v$. This induces a partitioning of the input matrices $S$ and $T$ into a two-level grid of submatrices; using the same wild-card notation as before, $S$ is partitioned into a $d \times d$ grid of $n/d \times n/d$ submatrices $S[\fmbdual{i}{\emptyind},\fmbdual{j}{\emptyind}]$ for $(i,j) \in [d]^2$, and each of these submatrices is further partitioned into an $n^{1/2} \times n^{1/2}$ grid of $n^{1/2}/d \times n^{1/2}/d$ submatrices $S[\fmbdual{i}{x},\fmbdual{j}{y}]$ for $x, y \in [n^{1/2}]$. The other input matrix $T$ is partitioned similarly; see Figure~\ref{fig:fast-mm}.

\begin{figure}[t]
    \centering
    \includegraphics[page=3]{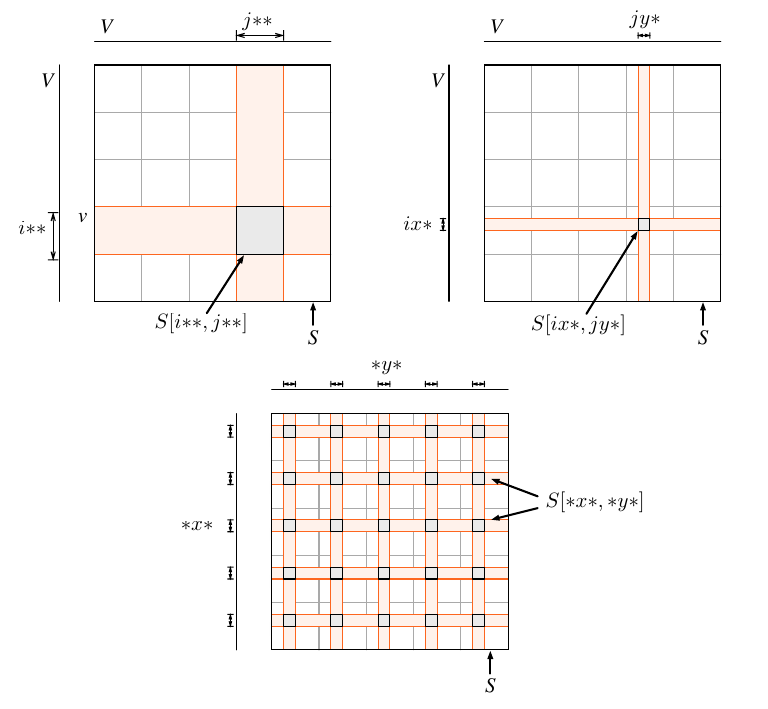}
    \caption{Fast matrix multiplication: partitioning schemes for matrix entries.}\label{fig:fast-mm}
\end{figure}

Finally, we give each node $v \in V$ a unique secondary label $\ell(v) = x_1x_2 \in [n^{1/2}]^2$; again, for concreteness we assume that $x_1x_2$ is the representation of $v$ in base-$n^{1/2}$ system, so this label can be computed from $v$ directly.

\paragraph{Overview.}
The basic idea of the fast distributed matrix multiplication is that we view the matrices $S$ and $T$ as $d\times d$ matrices $S'$ and $T'$ over the ring of $n/d \times n/d$ matrices, where
\[ S'_{ij} = S[\fmbdual{i}{\emptyind},\fmbdual{j}{\emptyind}]\,, \hspace{10mm}T'_{ij}=T[\fmbdual{i}{\emptyind},\fmbdual{j}{\emptyind}]\,, \hspace{10mm}i,j \in [d]\,,\]
which allows us to use \eqref{eq:bilin-1} and \eqref{eq:bilin-2} to compute the matrix product using the fixed bilinear algorithm; specifically, this reduces the $n \times n$ matrix product into $n$ instances of $n^{1-2/\mmeaux} \times n^{1-2/\mmeaux}$ matrix products, each of which is given to a different node. For the linear combination steps, we use a partitioning scheme where each node $v$ with secondary label $\ell(v) = x_1 x_2$ is responsible for $n^{1/2}/d \times n^{1/2}/d$ of the matrices involved in the computation.

\paragraph{Algorithm description.}
The algorithm computes the matrix product $P = ST$ as follows.
\begin{description}
    \item[Step 1: Distributing the entries.] Each node $v$ sends, for $x_2 \in [n^{1/2}]$, the submatrices $S[v,\fmbdual{\emptyind}{x_2}]$ and $T[v,\fmbdual{\emptyind}{x_2}]$ to the node $u$ with label $\ell(u) = v_2x_2$. Each submatrix has $n^{1/2}$ entries and there are $n^{1/2}$ recipients each receiving two submatrices, for a total of $2n$ messages per node.

    Dually, each node $u$ with label $\ell(u) = x_1x_2$ receives the submatrices $S[v,\fmbdual{\emptyind}{x_2}]$ and $T[v,\fmbdual{\emptyind}{x_2}]$ from the nodes $v=v_1 v_2 v_3$ with $v_2 = x_1$. In particular, node $u$ now has the submatrices $S[\fmbdual{\emptyind}{x_1},\fmbdual{\emptyind}{x_2}]$ and $T[\fmbdual{\emptyind}{x_1},\fmbdual{\emptyind}{x_2}]$. The total number of received messages is $2n$ per node.
    \item[Step 2: Linear combination of entries.] Each node $v$ with label $\ell(v) = x_1x_2$ computes for $w \in V$ the linear combinations
    \begin{align*}
        \hat{S}^{(w)}[\fmbsingle{x_1},\fmbsingle{x_2}] = & \sum_{(i,j) \in [d]^2} \alpha_{ijw} S[\fmbdual{i}{x_1},\fmbdual{j}{x_2}]\,,\hspace{10mm}\text{and}\\
        \hat{T}^{(w)}[\fmbsingle{x_1},\fmbsingle{x_2}] = & \sum_{(i,j) \in [d]^2} \beta_{ijw} T[\fmbdual{i}{x_1},\fmbdual{j}{x_2}]\,.
    \end{align*}
    The computation is performed entirely locally.
    \item[Step 3: Distributing the linear combinations.] Each node $v$ with label $\ell(v) = x_1x_2$ sends, for $w \in W$, the submatrices $\hat{S}^{(w)}[\fmbsingle{x_1},\fmbsingle{x_2}]$ and $\hat{T}^{(w)}[\fmbsingle{x_1},\fmbsingle{x_2}]$ to node $w$. Each submatrix has $(n^{1/2}/d)^2 = O(n^{1-2/\mmeaux})$ entries and there are $n$ recipients each receiving two submatrices, for a total of $O(n^{2-2/\mmeaux})$ messages per node.

    Dually, each node $w \in V$ receives the submatrices $\hat{S}^{(w)}[\fmbsingle{x_1},\fmbsingle{x_2}]$ and $\hat{T}^{(w)}[\fmbsingle{x_1},\fmbsingle{x_2}]$ from node $v \in V$ with label $\ell(v) = x_1x_2$. Node $u$ now has the matrices $\hat{S}^{(w)}$ and $\hat{T}^{(w)}$. The total number of received messages is $O(n^{2-2/\mmeaux})$ per node.
    \item[Step 4: Multiplication.] Node $w \in V$ computes the product $\hat{P}^{(w)} = \hat{S}^{(w)}\hat{T}^{(w)}$. The computation is performed entirely locally.
    \item[Step 5: Distributing the products.] Each node $w$ sends, for $x_1, x_2 \in [n^{1/2}]$, the submatrix $\hat{P}^{(w)}[\fmbsingle{x_1},\fmbsingle{x_2}]$ to node $v \in V$ with label $x_1x_2$. Each submatrix has $(n^{1/2}/d)^2 = O(n^{1-2/\mmeaux})$ entries and there are $n$ recipients, for a total of $O(n^{2-2/\mmeaux})$ messages sent by each node.

    Dually, each node $v \in V$ with label $\ell(v) = x_1x_2$ receives the submatrix $\hat{P}^{(w)}[\fmbsingle{x_1},\fmbsingle{x_2}]$ from each node $w \in V$. The total number of received messages is $O(n^{2-2/\mmeaux})$ per node.
    \item[Step 6: Linear combination of products.] Each node $v \in V$ with label $\ell(v) = x_1x_2$ computes for $i,j \in [d]$ the linear combination
        \[ P[\fmbdual{i}{x_1},\fmbdual{j}{x_2}] = \sum_{w \in V} \lambda_{ijw} \hat{P}^{(w)}[\fmbsingle{x_1},\fmbsingle{x_2}]\,.\]
    Node $v \in V$ now has the submatrix $P[\fmbdual{\emptyind}{x_1},\fmbdual{\emptyind}{x_2}]$. The computation is performed entirely locally.
    \item[Step 7: Assembling the product.] Each node $v \in V$ with label $\ell(v) = x_1x_2$ sends, for each node $u \in V$ with $u_2 = x_1$, the submatrix $P[u,\fmbdual{\emptyind}{x_2}]$ to the node $u$. Each submatrix has $n^{1/2}$ entries and there are $n^{1/2}$ recipients, for a total of $n$ messages sent by each node.

    Dually, each node $u \in V$ receives the submatrix $P[u,\fmbsingle{\emptyind x_2}]$ from the node $v$ with label $\ell(v) = u_2x_2$. Node $u$ now has the row $P[u,\emptyind]$ of the product matrix $P$. The total number of received messages is $n$ per node.
\end{description}

\paragraph{Analysis.}
The maximal number of messages sent or received by a node in the above steps is $O(n^{2-2/\mmeaux})$. Moreover, the communication pattern clearly does not depend on the input matrices, so the algorithm can be implemented in an oblivious way on the congested clique using the routing scheme of~\citet[Lemma 1]{tritri}; the running time is $O(n^{1 - 2/\mmeaux})$ rounds.


\section{Upper Bounds}


\subsection{Subgraph Detection and Counting}
\label{sec:subgraphs}

The subgraph detection and counting algorithms we present are mainly based on applying the fast matrix multiplication to the \emph{adjacency matrix $A$} of a graph $G = (V,E)$, defined as
\[
A_{uv} =
    \begin{cases}
    1 & \text{if } (u,v) \in E\,,\\
    0 & \text{if } (u,v) \notin E\,,
    \end{cases}
\]
where we assume that for undirected graphs edges $\{ u, v \} \in E$ are oriented both ways.

\paragraph{Counting triangles and 4-cycles.}
For counting triangles, that is, $3$-cycles, we use a technique first observed by \citet{itai1978finding}. That is, in an undirected graph with adjacency matrix $A$, the number of triangles is known to be
$ \frac{1}{6} \trace(A^3)$, where the \emph{trace $\trace(S)$} of a matrix $S$ is the sum of its diagonal entries $S_{uu}$. Similarly, for directed graphs, the number of triangles is
$ \frac{1}{3} \trace(A^3)$.

\citet{alon1997finding} generalise the above formula to counting undirected and directed $k$-cycles for small $k$. For example, the number of $4$-cycles in an undirected graph is given by
\[ \frac{1}{8}\Bigl[ \trace(A^4) - \sum_{v \in V} \Bigr( 2 (\deg(v))^2 - \deg(v) \Bigr) \Bigr]\,.\]
Likewise, if $G$ is a loopless directed graph and we denote for $v \in V$ by $\delta(v)$
the number of nodes $u \in V$ such that $\{ (u,v), (v,u) \} \subseteq E$,
then the number of directed $4$-cycles in $G$ is
\[ \frac{1}{4}\Bigl[ \trace(A^4) - \sum_{v \in V} \Bigr( 2 (\delta(v))^2 - \delta(v) \Bigr) \Bigr]\,.\]
Combining these observations with \theoremref{thm:mm}, we immediately obtain \corollaryref{cor:triangles}:

\thmtriangles*

We note that similar trace formulas exists for counting $k$-cycles for $k \in \{ 5, 6, 7 \}$,
requiring only computation of small powers of $A$ and local information.
We omit the detailed discussion of these in the context of the congested clique;
see \citet{alon1997finding} for details.

\paragraph{Detecting $k$-cycles.}
For detection of $k$-cycles 
we leverage the \emph{colour-coding} techniques of \citet{alon1995color} in addition to the matrix multiplication.
Again, the distributed algorithm is a straightforward adaptation of a centralised one.

Fix a constant $k\in \mathbb{N}$.
Let $c \colon V \to [k]$ be a labelling (or colouring) of the nodes by $k$ colours,
such that node $v$ knows its colour $c(v)$;
it should be stressed here that the colouring need not to be a proper colouring in the sense of the graph colouring problem.
As a first step, we consider the problem of finding a \emph{colourful $k$-cycle}, that is, a $k$-cycle such that each colour occurs exactly once on the cycle. We present the details assuming that the graph $G$ is directed, but the technique works in an identical way for undirected graphs.

\begin{lemma}\label{lemma:cccc}
Given a graph $G = (V,E)$ and a colouring $c \colon V \to [k]$, a colourful $k$-cycle can be detected in $O\bigl( 3^k n^{\mme} \bigr)$ rounds.
\end{lemma}

\begin{proof}
For each subset of colours $X \subseteq [k]$, let $C^{(X)}$ be a Boolean matrix such that $C^{(X)}_{uv} = 1$ if there is a path of length $\card{X}-1$ from $u$ to $v$ containing exactly one node of each colour from $X$,
and $C^{(X)}_{uv} = 0$ otherwise.
For a singleton set $\{ i \} \subseteq [k]$, the matrix $C^{(\{ i \}) }$
contains $1$ only on the main diagonal, and only for nodes $v$ with $c(v) = i$;
hence, node $v$ can locally compute the row $v$ of the matrix from its colour.
For a non-singleton colour set $X$, we have that
\begin{equation}\label{eq:colourful-cycle-recurrence}
    C^{(X)} = \bigvee_{\substack{ Y \subseteq X \\ \card{Y} = \lceil \card{X}/2 \rceil}} C^{(Y)} A C^{(X \setminus Y)}\,,
\end{equation}
where the products are computed over the Boolean semiring and $\lor$ denotes element-wise logical or.
Thus, we can compute $C^{(X)}$ for all $X \subseteq [k]$ by applying \eqref{eq:colourful-cycle-recurrence} recursively;
there is a colourful $k$-cycle in $G$ if and only if
there is a pair of nodes $u,v\in V$ such that $C^{([k])}_{uv} = 1$ and $(v,u) \in E$.

To leverage fast matrix multiplication, we simply perform the operations stated in \eqref{eq:colourful-cycle-recurrence} over the ring $\mathbb{Z}$ and observe that an entry of the resulting matrix is non-zero if and only if the corresponding entry of $C^{(X)}$ is non-zero.
The application of (\ref{eq:colourful-cycle-recurrence}) needs two matrix multiplications for each pair $(Y,X)$
with $Y \subseteq [k]$ and $\card{Y} = \lceil \card{X}/2 \rceil = \lceil k/2\rceil$.
The number of such pairs is bounded by $3^k$; to see this, note that the set $\{ (Y,X) \colon Y \subseteq X \subseteq [k] \}$ can be identified with the set $\{ 0, 1, 2 \}^k$ of trinary strings of length $k$ via the bijection $w_1 w_2 \dotsc w_k \mapsto (\{ i \colon w_i = 0 \}, \{ i \colon w_i \le 1 \})$, and the set $\{ 0, 1, 2 \}^k$ has size exactly $3^k$.
Thus, the total number of matrix multiplications used is at most $O(3^k)$.
\end{proof}

We can now use \lemmaref{lemma:cccc} to prove \theoremref{thm:k-cycles}; while we cannot directly construct a suitable colouring from scratch for an uncoloured graph, we can try an exponential in $k$ number of colourings to find a suitable one.

\thmkcycles*

\begin{proof}
To apply \lemmaref{lemma:cccc}, we first have to obtain a colouring $c \colon V \to [k]$
that assigns each colour once to at least one $k$-cycle in $G$,
assuming that one exists. If we pick a colour $c(v) \in [k]$ for each node uniformly at random,
then for any $k$-cycle $C$ in $G$,
the probability that $C$ is colourful in the colouring $c$ is
$k! / k^k < e^{-k}$.
Thus, by picking $e^k\log n$ uniformly random colourings and applying \lemmaref{lemma:cccc}
to each of them, we find a $k$-cycle with high probability if one exists.

%

This algorithm can also be derandomised using standard techniques. A \emph{$k$-perfect family of hash functions} $\mathcal{H}$ is a collection of functions $h \colon V \to [k]$ such that for each $U \subseteq V$ with $\card{U} = k$, there is at least one $h \in \mathcal{H}$ such that $h$ assigns a distinct colour to each node in $U$. There are known constructions that give such families $\mathcal{H}$ with $\card{\mathcal{H}} = 2^{O(k)} \log n$ and these can be efficiently constructed~\cite{alon1995color}; thus, it suffices to take such an $\mathcal{H}$ and apply \lemmaref{lemma:cccc} for each colouring $h \in \mathcal{H}$.
\end{proof}

\paragraph{Detecting 4-cycles.}

We have seen how to \emph{count} $4$-cycles with the help of matrix multiplication in $O(n^\mme)$ rounds. We now show how to \emph{detect} $4$-cycles in $O(1)$ rounds. The algorithm does not make direct use of matrix multiplication algorithms. However, the key part of the algorithm can be interpreted as an efficient routine for sparse matrix multiplication, under a specific definition of sparseness.

Let
\[
    P(X,Y,Z) = \{ (x,y,z) : x \in X, y \in Y, z \in Z, \{x,y\} \in E, \{y,z\} \in E \}
\]
consist of all distinct $2$-walks (paths of length $2$) from $X$ through $Y$ to $Z$. We will use again the shorthand notation $v$ for $\{v\}$ and $*$ for $V$; for example, $P(x,*,*)$ consists of all walks of length $2$ from node $x$. There exists a $4$-cycle if and only if $|P(x,*,z)| \ge 2$ for some $x \ne z$.

On a high level, the algorithm proceeds as follows.
\begin{enumerate}
    \item Each node $x$ computes $|P(x,*,*)|$. If $|P(x,*,*)| \ge 2n-1$, then there has to be some $z \ne x$ such that $|P(x,*,z)| \ge 2$, which implies that there exists a $4$-cycle, and the algorithm stops.
    \item Otherwise, each node $x$ finds $P(x,*,*)$ and checks if there exists some $z \ne x$ such that $|P(x,*,z)| \ge 2$.
\end{enumerate}
The first phase is easy to implement in $O(1)$ rounds. The key idea is that if the algorithm does not stop in the first phase, then the total volume of $P(*,*,*)$ is sufficiently small so that we can afford to gather $P(x,*,*)$ for each node $x$ in $O(1)$ rounds.

We now present the algorithm in more detail. We write $N(x)$ for the neighbours of node $x$. To implement the first phase, it is sufficient for each node $y$ to broadcast $\deg(y) = |N(y)|$ to all other nodes; we have
\[
    |P(x,*,*)| = \sum_{y \in N(x)} \deg(y).
\]

Now let us explain the second phase. Each node $y$ is already aware of $N(y)$ and hence it can construct $P(*,y,*) = N(y) \times \{y\} \times N(y)$. Our goal is to distribute the set of all $2$-walks
\[
    \bigcup_y P(*,y,*) = P(*,*,*) = \bigcup_x P(x,*,*)
\]
so that each node $x$ will know $P(x,*,*)$.

In the second phase, we have
\[
    \sum_y \deg(y)^2
    = \sum_y |P(*,y,*)|
    = \sum_x |P(x,*,*)|
    < 2n^2.
\]
Using this bound, we obtain the following lemma.
\begin{lemma}\label{lem:4cAB}
    It is possible to find sets $A(y)$ and $B(y)$ for each $y \in V$ such that the following holds:
    \begin{itemize}
        \item $A(y) \subseteq V$, $B(y) \subseteq V$, and $|A(y)| = |B(y)| \ge \deg(y)/8$,
        \item the tiles $A(y) \times B(y)$ are disjoint subsets of the square $V \times V$.
    \end{itemize}
    Moreover, this can be done in $O(1)$ rounds in the congested clique.
\end{lemma}
\begin{proof}
    Let $f(y)$ be $\deg(y)/4$ rounded down to the nearest power of $2$, and let $k$ be $n$ rounded down to the nearest power of $2$. We have $\sum_y f(y)^2 \le \sum \deg(y)^2/16 < n^2/8 < k^2$. Now it is easy to place the tiles of dimensions $f(y) \times f(y)$ inside a square of dimensions $k \times k$ without any overlap with the following iterative procedure:
    \begin{itemize}
        \item Before step $i = 1,2,\dotsc$, we have partitioned the square in sub-squares of dimensions $k/2^{i-1} \times k/2^{i-1}$, and each sub-square is either completely full or completely empty.
        \item During step $i$, we divide each sub-square in $4$ parts, and fill empty squares with tiles of dimensions $f(y) = k/2^i$.
        \item After step $i$, we have partitioned the square in sub-squares of dimensions $k/2^i \times k/2^i$, and each sub-square is either completely full or completely empty.
    \end{itemize}
    This way we have allocated disjoint tiles $A(y) \times B(y) \subseteq [k] \times [k] \subseteq V \times V$ for each $y$, with $|A(y)| = |B(y)| = f(y) \ge \deg(y)/8$.

    To implement this in the congested clique model, it is sufficient that each $y$ broadcasts $\deg(y)$ to all other nodes, and then all nodes follow the above procedure to compute $A(y)$ and $B(y)$ locally.
\end{proof}

Now we will use the tiles $A(y) \times B(y)$ to implement the second phase of $4$-cycle detection. For convenience, we will use the following notation for each $y \in Y$:
\begin{itemize}
    \item The sets $N_A(y,a)$ where $a \in A(y)$ form a partition of $N(y)$ with $|N_A(y,a)| \le 8$.
    \item The sets $N_B(y,b)$ where $b \in B(y)$ form a partition of $N(y)$ with $|N_B(y,b)| \le 8$.
\end{itemize}
Note that we can assume that $A(y)$ and $B(y)$ are globally known by Lemma~\ref{lem:4cAB}. Hence a node can compute $N_A(y,a)$ and $N_B(y,b)$ if it knows $N(y)$.

\begin{figure}
    \centering
    \includegraphics[page=1]{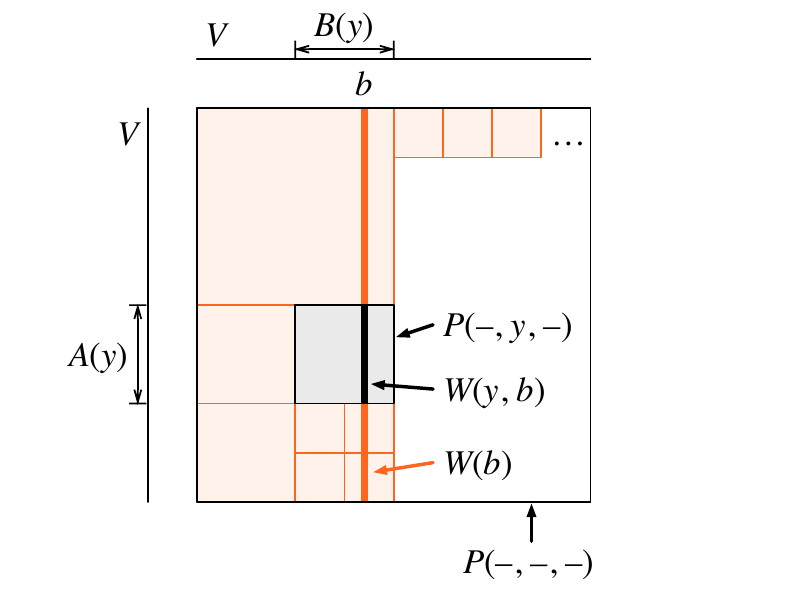}
    \caption{$4$-cycle detection: how $P(*,*,*)$ is partitioned among the nodes.}\label{fig:4cycle}
\end{figure}

With this notation, the algorithm proceeds as follows (see Figure~\ref{fig:4cycle}):
\begin{enumerate}[parsep=1.5ex]
    \item For all $y \in V$ and $a \in A(y)$, node $y$ sends $N_A(y,a)$ to $a$.

    This step can be implemented in $O(1)$ rounds.
    \item For each $y$ and each pair $(a,b) \in A(y) \times B(y)$, node $a$ sends $N_A(y,a)$ to $b$.

    Note that for each $(a,b)$ there is at most one $y$ such that $(a,b) \in A(y) \times B(y)$; hence over each edge we send only $O(1)$ words. Therefore this step can be implemented in $O(1)$ rounds.

    \item At this point, each $b \in V$ has received a copy of $N(y)$ for all $y$ with $b \in B(y)$. Node $b$ computes
    \[
        W(y,b) = N(y) \times \{y\} \times N_B(y,b), \qquad W(b) = \bigcup_{y: b \in B(y)} W(y,b).
    \]

    This is local computation; it takes $0$ rounds.
\end{enumerate}
We now give a lemma that captures the key properties of the algorithm.
\begin{lemma}
    The sets $W(b)$ form a partition of $P(*,*,*)$. Moreover, for each $b$ we have $|W(b)| = O(n)$.
\end{lemma}
\begin{proof}
    For the first claim, observe that the sets $P(*,y,*)$ for $y \in V$ form a partition of $P(*,*,*)$, the sets $W(y,b)$ for $b \in B(y)$ form a partition of $P(*,y,*)$, and each set $W(y,b)$ is part of exactly one $W(b)$.

    For the second claim, let $Y$ consist of all $y \in V$ with $b \in B(y)$. As the tiles $A(y) \times B(y)$ are disjoint for all $y \in Y$, and all $y \in Y$ have the common value $b \in B(y)$, it has to hold that the sets $A(y)$ are disjoint subsets of $V$ for all $y \in Y$. Therefore
    \[
        \sum_{y \in Y} |N(y)|
        = \sum_{y \in Y} \deg(y)
        \le \sum_{y \in Y} 8|A(y)|
        \le 8|V| = 8n.
    \]
    With $|N_B(y)| \le 8$ we get
    \[
        |W(b)| = \sum_{y \in Y} |W(y,b)| \le 8 \sum_{y \in Y} |N(y)| \le 64n. \qedhere
    \]
\end{proof}

Now we are almost done: we have distributed $P(*,*,*)$ evenly among $V$ so that each node only holds $O(n)$ elements. Finally, we use the dynamic routing scheme \cite{lenzen2013optimal} to gather $P(x,*,*)$ at each node $x \in V$; here each node needs to send $O(n)$ words and receive $O(n)$ words, and the running time is therefore $O(1)$ rounds. In conclusion, we can implement both phases of $4$-cycle detection in $O(1)$ rounds.

\thmfourcycles*


\subsection{Girth}
\label{sec:girth}

\paragraph{Undirected girth.}
Recall that the \emph{girth} $g$ of an undirected unweighted graph $G = (V,E)$ is
the length of the shortest cycle in $G$. To compute the girth in the congested clique model,
we leverage the fast cycle detection algorithm and the following lemma giving a trade-off
between the girth and the number of edges.
A similar approach of bounding from above the number of edges of a graph that contains no copies of some given subgraph was taken by \citet{drucker13}.


\begin{lemma}[{\cite[pp. 362--363]{Matousek02_geometry}}]\label{lemma:girth}
A graph with girth $g$ has at most $n^{1+1/\lfloor(g-1)/2\rfloor}+n$ edges.
\end{lemma}

If the graph is dense, then by the above lemma it must have small girth
and we can use fast cycle detection to compute it; otherwise, the graph is sparse and we can learn the complete graph structure.

\begin{theorem}
For undirected graphs, the girth can be computed in $\tilde{O}(n^{\mme})$ rounds (or in $n^{o(1)}$ rounds, 
if $\mme = 0$).
\end{theorem}

\begin{proof}
Assume for now that $\mme > 0$, and fix 
$\ell = \lceil 2 + 2/\mme\rceil$. 
Each node collects all graph degrees and computes the total number of edges.
If there are at most $n^{1+1/\lfloor \ell /2 \rfloor}+n = O(n^{1 + \mme})$ edges,
we can collect full information about the graph structure to all nodes in $O(n^{\mme})$ rounds using an algorithm of \citet{tritri}, and each node can then compute the girth locally.

Otherwise, by \lemmaref{lemma:girth}, the graph has girth at most $\ell$. Thus, for $k = 3, 4, \dotsc, \ell$, we try to find a $k$-cycle using \theoremref{thm:k-cycles}, in $\ell \cdot 2^{O(\ell )}n^{\mme } \log n = \tilde{O}(n^\mme)$ rounds. When such a cycle is found for some $k$, we stop and return $k$ as the girth.

Finally, if $\mme=0$, we pick $\ell = \log\log n$, and both cases take $n^{o(1)}$ rounds.
\end{proof}

\paragraph{Directed girth.}
For a directed graph, the girth is defined as the length of the shortest directed cycle; the main difference is that directed girth can be $1$ or $2$. While the trade-off of \lemmaref{lemma:girth} cannot be used for directed graphs, we can use a simpler technique of \citet{itai1978finding}.

Let $G=(V,E)$ be a directed graph; we can assume that there are no self-loops in $G$, as otherwise girth is $1$ and we can detect this with local computation. Let $B^{(i)}$ be a Boolean matrix defined as
\[
B^{(i)}_{uv} =
    \begin{cases}
    1 & \text{if there is a path of length $\ell$ from $u$ to $v$ for $1 \le \ell \le i$,}\\
    0 & \text{otherwise.}
    \end{cases}
\]
Clearly, we have that $B^{(1)} = A$. Moreover, if $i = j + k $, we have
\begin{equation}\label{eq:dir-girth}
B^{(i)} = \bigl(B^{(j)} B^{(k)}\bigr) \lor A\,,
\end{equation}
where the matrix product is over the Boolean semiring and $\lor$ denotes element-wise logical or.

\begin{corollary}
For directed graphs, the girth can be computed in $\tilde{O}(n^{\mme})$ rounds.
\end{corollary}

\begin{proof}
It suffices to find smallest $\ell$ such that there is $v \in V$ with $B^{(\ell)}_{vv} = 1$; clearly $\ell$ is then the girth of graph $G$. We first compute $A = B^{(1)}, B^{(2)}, B^{(4)}, B^{(8)}, \dotsc$ using (\ref{eq:dir-girth}) with $j=k=i/2$ until we find $i$ such that $B^{(i)}_{vv} = 1$ for some $v \in V$. We then know that the girth is between $i$ and $i/2$; we can perform binary search on this interval to find the girth, using (\ref{eq:dir-girth}) to evaluate the intermediate matrices. This requires $O(\log n)$ calls to the matrix multiplication algorithm.
\end{proof} 


\subsection{Routing and Shortest Paths}
\label{sec:routing}

In this section, we present algorithms for variants of the all-pairs shortest paths (APSP) problem. In the congested clique model, the local input for a node $u \in V$ in the APSP problem is a vector containing the local edge weights $W(u,v)$ for $v \in V$. The output for $u \in V$ is the actual shortest path distances $d(u,v)$ for each other node $v \in V$, along with the \emph{routing table} entries $R[u,v]$, where each entry $R[u,v] = w \in V$ is a node such that $(u,w) \in V$ and $w$ lies on a shortest path on from $u$ to $w$. For convenience, we use the same notation for directed and undirected graphs, assume $W(u,v) = \infty$ if $(u,v) \notin E$, and for unweighted graphs, we set $W(u,v) = 1$ for each $(u,v) \in E$.

For a graph $G = (V,E)$ with edge weights $W$, we define the \emph{weight matrix $W$} as
\[
W_{uv} =
    \begin{cases}
    W(u,v) & \text{if } u \not= v \,,\\
    0 & \text{if } u = v\,.
    \end{cases}
\]
Our APSP algorithms are mostly based on the manipulation of the weight matrix $W$ and the adjacency matrix $A$, as defined in \sectionref{sec:subgraphs}.

\paragraph{Distance product and iterated squaring.}
Matrix multiplication can be used to compute the shortest path distances via \emph{iterated squaring} of the weight matrix over the min-plus semiring~\cite{furman1970application, Munro197156, fm1971boolean}. That is, the matrix product is the \emph{distance product}, also known as the \emph{min-plus product} or \emph{tropical product}, defined as
\[ (S \star T)_{uv} = \min_{w} \bigl( S_{uw} + T_{wv} \bigr)\,.\]
Given a graph $G = (V,E)$ with weight matrix $W$, the $n^\text{th}$ distance product power $W^{n}$ gives the actual distances in $G$ as $d(v,u) = W^{ n}_{vu}$. Computing $W^{n}$ can be done with $\lceil \log n \rceil$ distance products by iteratively squaring $W$, that is, we compute
\[
   W^{2} = W \star  W\,, \hspace{5mm} W^{4} = W^{2} \star  W^{2}\,,\hspace{5mm}\dotsc,\hspace{5mm} W^{n} = W^{ n/2} \star W^{n/2}\,.
\]
Combining this observation with the semiring algorithm from \theoremref{thm:mm}, we immediately obtain a simple APSP algorithm for the congested clique.

\corrapspsemi*

The subsequent APSP algorithms we discuss in this section are, for the most part, similarly based on the iterated squaring of the weight matrix; the main difference is that we replace the semiring matrix multiplication with distance product algorithms derived from the fast matrix multiplication algorithm.

\paragraph{Constructing routing tables.}
The iterated squaring algorithm of \corollaryref{thm:apsp-semiring} can be adapted to also compute a routing table $R$ as follows. Assume that our distance product algorithm also provides for the distance product $S \star T$ a \emph{witness matrix $Q$} such that if $Q_{uv} = w$, then $(S \star T)_{uv} = S_{uw} + T_{wv}$. With this information, we can compute the routing table $R$ during the iterated squaring algorithm; when we compute the product $W^{2i} = W^{i} \star W^{i}$, we also obtain a witness matrix $Q$, and update the routing table by setting
\[ R[u,v] = R[u, Q_{uv}] \]
for each $u,v\in V$ with $W^{2i}_{uv} < W^{i}_{uv}$.

The semiring matrix multiplication can be easily modified to produce witnesses, but for the subsequent distance product algorithms based on fast matrix multiplication this is not directly possible. However, we can apply known techniques from the centralised setting to obtain witnesses also in these cases~\cite{alon1996derandomization,Seidel1995400,zwick2002all}; we refer to \sectionref{sec:witnesses} for details.

\paragraph{Unweighted undirected APSP.}
In the case of unweighted undirected graphs, we can obtain exact all-pairs shortest paths via a technique of \citet{Seidel1995400}. Specifically, let $G = (V,E)$ an unweighted undirected graph with adjacency matrix $A$; the \emph{$k^\text{th}$ power $G^k$} of $G$ is a graph with node set $V$ and edge set $\{ \{ u, v \} \colon d(u,v) \le k \}$. In particular, the \emph{square graph} $G^2$ can be constructed in $O(n^\mme)$ rounds from $G$, as the adjacency matrix of $G^2$ is $A^2 \lor A$, where the product is over the Boolean semiring and $\lor$ denotes element-wise logical or.

The following lemma of Seidel allows us to compute distances in $G$ if we already know distances in the square graph $G^2$; to avoid ambiguity, we write in this subsection $d_G(u,v)$ for the distances in a graph $G$.

\begin{lemma}[\cite{Seidel1995400}]\label{lemma:seidel}
Let $G = (V,E)$ be an unweighted undirected graph with adjacency matrix $A$, and let $D$ be a distance matrix for $G^2$, that is, a matrix with the entries $D_{uv} = d_{G^2}(u,v)$. Let $S = DA$, where the product is computed over integers. We have that
 \begin{equation*}
      d_{G}(u,v) =
      \begin{cases}
           2 d_{G^2}(u,v) & \text{if $S_{uv} \ge d_{G^2}(u,v)\deg_{G}(v)$, and}\\
          2 d_{G^2}(u,v) - 1 &  \text{if $S_{uv} < d_{G^2}(u,v)\deg_{G}(v)$.}
      \end{cases}
 \end{equation*}
\end{lemma}

We can now recover all-pairs shortest distances in an undirected unweighted graph by recursively applying \lemmaref{lemma:seidel}.

\corseidel*

\begin{proof}
Let $G = (V,E)$ be an unweighted undirected graph with adjacency matrix $A$. We first compute the adjacency matrix for $G^2$; as noted above, this can be done in $O(n^\mme)$ rounds. There are now two cases to consider.
\begin{enumerate}
    \item If $G = G^2$, then $d_G(u,v) = 1$ if $u$ and $v$ are adjacent in $G$, and $d_G(u,v) = \infty$ otherwise; thus, we are done.
    \item Otherwise, we compute all-pairs shortest path distances in the graph $G^2$; since we have already constructed the adjacency matrix for $G^2$, we can do the distance computation in $G^2$ by recursively calling this algorithm with input graph $G^2$. Then, we construct the matrix $D$ with entries $D_{uv} = d_{G^2}(u,v)$ as in \lemmaref{lemma:seidel} and compute $S = DA$. We can recover distances in $G$ using \lemmaref{lemma:seidel}, as each node can transmit their degree in $G$ to each other node in a single round and then check the conditions of the lemma locally.
\end{enumerate}
The recursion terminates in $O(\log n)$ calls, as the graph $G^n$ consists of disjoint cliques.
\end{proof}

\paragraph{Weighted APSP with small weights.}
By embedding the distance product of two matrices into a suitable ring, we can use fast ring matrix multiplication to compute all-pairs shortest distances~\cite{Yuval76_apsp}; however, this is only practical for very small weights, as the ring embedding exponentially increases the amount of bits required to transmit the matrix entries. The following lemma encapsulates this idea.

\begin{lemma}\label{lemma:min-sum-emb}
Given $n \times n$ matrices $S$ and $T$ with entries in $\{ 0, 1, \dotsc, M \} \cup \{ \infty \}$, we can compute the distance product $S \star T$
in $O(M n^{\mme})$ rounds.
\end{lemma}

\begin{proof}
We construct matrices $S^*$ and $T^*$ by replacing each matrix entry $w$ with $X^w$, where $X$ is a formal variable; values  $\infty$ are replaced by $0$. We then compute the product $S^* \cdot T^*$ over the polynomial ring $\Z[X]$; all polynomials involved in the computation have degree at most $2M$ and their coefficients are integers of absolute value at most $n^{O(1)}$, so this computation can be done in $O(M n^{\mme})$ rounds. Finally, we can recover each matrix entry $(S \star T)_{uv}$ in the original distance product by taking the degree of the lowest-degree monomial in $(S^* \cdot T^*)_{uv}$.
\end{proof}


Using iterated squaring in combination with \lemmaref{lemma:min-sum-emb}, we can compute all-pairs shortest paths up to a small distance $M$ quickly; that is, we want to compute a matrix $B$ such that
\begin{equation*}
B_{uv} =
\begin{cases}
d(u,v) & \text{if $d(u,v) \le M$, and}\\
\infty & \text{if $d(u,v) > M$.}
\end{cases}
\end{equation*}
This can be done by replacing all weights over $M$ with $\infty$ before each squaring operation to ensure that we do not operate with too large values, giving us the following lemma.

\begin{lemma}\label{lemma:apsp-poly}
Given a directed, weighted graph with non-negative integer weights, we can compute all-pairs shortest paths up to distance $M$ in $O(M n^{\mme})$ rounds.
\end{lemma}

The above lemma can be used to compute all-pairs shortest paths quickly assuming that the \emph{weighted diameter} of the graph is small; recall that the weighted diameter of a weighted graph is the maximum distance between any pair of nodes.

\cordiamapsp*

\begin{proof}
If we know that the weighted diameter is $U$, we can simply apply \lemmaref{lemma:apsp-poly} with $M = U$. However, if we do not know $U$ beforehand, we can (1)~first compute the reachability matrix of the graph from the unweighted adjacency matrix, (2)~guess $U = 1$ and compute all-pairs shortest paths up to distance $U$, and (3)~check if we obtained distances for all pairs that are reachable according to the reachability matrix; if not, then we double our guess for $U$ and repeat steps (2) and (3).
\end{proof}

\paragraph{Approximate weighted APSP.}
We can leverage the above result and a rounding technique to obtain a fast $(1+o(1))$-approximation algorithm for the weighted directed APSP problem. Similar rounding-based approaches were previously used by \citet{zwick2002all} in a centralised setting and by \citet{nanongkai14} in the distributed setting; however, the idea can be traced back much further~\cite{raghavan85}.

We first consider the computation of a $(1+\delta)$-approximate distance product over integers for a given $\delta>0$; the following lemma is an analogue of one given by \citet{zwick2002all} in a centralised setting.

\begin{lemma}\label{lemma:approx-dp}
Given $n \times n$ matrices $S$ and $T$ with entries in $\{ 0, 1, \dotsc, M \} \cup \{ \infty \}$, we can compute a matrix $\tilde{P}$ satisfying
\[ P_{uv} \le \tilde{P}_{uv}\le (1 + \delta)P_{uv} \hspace{10mm}\text{for $u,v \in V$}\,, \]
where $P = S \star T$ is the distance product of $S$ and $T$, in $O\bigl(n^{\mme} (\log_{1+\delta} M) / \delta \bigr)$ rounds.
\end{lemma}

\begin{proof}
For $i\in \{0,\ldots,\lceil \log_{1+\delta} M\rceil\}$, let $S^{(i)}$ be the matrix defined as
\begin{equation*}
S^{(i)}_{uv} =
\begin{cases}
    \lceil S_{uv} / (1 + \delta)^i \rceil & \text{if $S_{uv} \le 2(1+\delta)^{i+1}/\delta$, and}\\
    \infty                                & \text{otherwise,}
\end{cases}
\end{equation*}
and let $T^{(i)}$ be defined similarly for $T$. Furthermore, let us define $P^{(i)} = S^{(i)} \star T^{(i)}$. We now claim that selecting
\[ \tilde{P}_{uv} = \min_{i} \bigl\{ \lfloor (1 + \delta)^{i} P^{(i)}_{uv}\rfloor \bigr\} \]
gives a matrix $\tilde{P}$ with the desired properties.

It follows directly from the definitions that $P_{uv} \le \tilde{P}_{uv}$, so it remains to prove the other inequality. Thus, let us fix $u,v \in V$, and let $w \in V$ be such that
\[ P_{uv} = S_{uw} + T_{wv}\,.\]
Finally, let $j = \lfloor \log_{1+\delta} (\delta P_{uv}/2)\rfloor$. The choice of $j$ means that $P_{uv} \le 2(1+\delta)^{j+1}/\delta$; since $S_{uw}$ and $T_{wv}$ are bounded from above by $P_{uv}$, the entries $S^{(j)}_{uw}$ and $T^{(j)}_{wv}$ are finite. Furthermore, we have
\[ (1 + \delta)^{j} S^{(j)}_{uw} \le S_{uw} + (1 + \delta)^j\,,\hspace{10mm} (1 + \delta)^{j} T^{(i)}_{wv} \le T_{wv} + (1 + \delta)^j\,, \]
and therefore
\begin{align*}
(1+ \delta)^{j} P^{(j)}_{uv} & \le (1+ \delta)^{j}\bigl( S^{(j)}_{uw} + T^{(j)}_{wv}\bigr) \\
                         & \le S_{uw} + T_{wv} + 2 (1+ \delta)^j\\
                         & \le P_{uv} + \delta P_{uv} = (1 + \delta)P_{uv}\,.
\end{align*}
Finally, we have $\tilde{P}_{uv} \le \lfloor (1+ \delta)^{j} P^{(j)}_{uv} \rfloor \le (1+\delta) P_{uv}$.

To see that we can compute the matrix $\tilde{P}$ in the claimed time, we first note that each of the matrices $S^{(i)}$ and $T^{(i)}$ can be constructed locally by the nodes. The product $P^{(i)} = S^{(i)} \star T^{(i)}$ can be computed in $O(n^{\mme}/\delta)$ rounds for a single index $i$ by \lemmaref{lemma:min-sum-emb}, as the entries of $S^{(i)}$ and $T^{(i)}$ are integers bounded from above by $O(1/\delta)$; this is repeated for each index $i$, and the number of iterations is thus $O(\log_{1 + \delta} M )$. Finally, the matrix $\tilde{P}$ can be constructed from matrices $P^{(i)}$ locally.
\end{proof}

Using \lemmaref{lemma:approx-dp}, we obtain a $(1+o(1))$-approximate APSP algorithm.

\thmapxapsp*

\begin{proof}
Let $G = (V,E)$ be a directed weighted graph with edge weights in $\{ 0, 1, \dotsc, M \}$, where $M = 2^{n^{o(1)}}$.
To compute the approximate shortest paths,
we apply iterated squaring over the min-plus semiring to the weight matrix $W$ of $G$,
but use the approximate distance product algorithm of \lemmaref{lemma:approx-dp} to compute the products.
After $\lceil \log n \rceil$ iterations, we obtain a matrix $\tilde{D}$;
by induction we have
\[ d(u,v) \le \tilde{D}_{uv} \le (1+\delta)^{\lceil \log n \rceil}d(u,v)\hspace{10mm}\text{for $u,v \in V$}\,.\]
Selecting $\delta = o(1/\log n)$, this gives a $(1 + o(1))$-approximation for the shortest distances.

To analyse the running time, we observe that we call the algorithm of \lemmaref{lemma:approx-dp} $\lceil \log n \rceil$ times; as the maximum distance between nodes in $G$ is $nM = 2^{n^{o(1)}}$, the running time of each call is bounded by
\[ O\left(\frac{n^{\mme} \log_{1+\delta}(nM)}{\delta} \right) = O\left(\frac{n^{\mme + o(1)}}{\delta \log (1+\delta)}\right)\,. \]
For sufficiently small $\delta$, we have $1 / \bigl(\delta \log (1+\delta)\bigr) = O(1 / \delta^2)$. Thus, for, e.g., $\delta = 1/\log^2 n = o(1/\log n)$, the total running time is $O(n^{\mme + o(1)})$, as the polylogarithmic factors are subsumed by $n^{o(1)}$.
\end{proof}


\subsection{Witness Detection for Distance Product}
\label{sec:witnesses}

\paragraph{Witness problem for the distance product.}
As noted in \sectionref{sec:routing}, to recover the routing table in the APSP algorithms based on fast matrix multiplication in addition to computing the shortest path lengths, we need the ability to compute a \emph{witness matrix} for the distance product $S \star T$. That is, we need to find a matrix $Q$ such that if $Q_{uv} = w$, then $(S \star T)_{uv} = S_{uw} + T_{wv}$; in this case, the index $w$ is called a \emph{witness} for the pair $(u,v)$.

While one can easily modify the semiring matrix multiplication algorithm to provide witnesses, this is not directly possible with the fast matrix multiplication algorithms.  However, known techniques from centralised algorithms~\cite{alon1996derandomization,Seidel1995400,zwick2002all} can be adapted to the congested clique to bridge this gap.

\begin{lemma}\label{lemma:witness}
If we can compute the distance product for two $n \times n$ matrices $S$ and $T$ in $M$ rounds, we can also find a witness matrix for $S \star T$ in $M \operatorname{polylog}(n)$ rounds.
\end{lemma}

The rest of this section outlines the proof of this lemma. While we have stated it for the distance product, it should be noted that the same techniques also work for the Boolean semiring matrix product.

\paragraph{Preliminaries.}
For matrix $S$ and index subsets $U, W \subseteq V$, we define the matrix $S(U,W)$ as
\[
S(U,W)_{uw} =
    \begin{cases}
    S_{uw} & \text{if $u \in U$ and $w \in W$,} \\
    \infty & \text{otherwise}.
    \end{cases}
\]
That is, we set all rows and columns not indexed by $U$ and $W$ to $\infty$. As before, we use $\emptyind$ as a shorthand for the whole index set $V$.

\paragraph{Finding unique witnesses.}
As a first step, we compute witnesses for all $(u,v)$ that have a \emph{unique} witness, that is, there is exactly one index $w$ such that $(S \star T)[u,v] = S[u,w] + T[w,v]$. To construct a candidate witness matrix $Q$, let $V^{(i)} \subseteq V$ be the set of indices $v$ such that bit $i$ in the binary presentation of $v$ is $1$. For $i = 1, 2, \dotsc, \lceil \log n \rceil$, we compute the distance product $P^{(i)} = S(\emptyind,V_i) \star T(V_i,\emptyind)$ If $P^{(i)}_{uv} = (S \star T)_{uv}$, then we set the $i^\text{th}$ bit of $Q_{uv}$ to $1$, and otherwise we set it to $0$.

If there is a unique witness for $(u,v)$, then $Q_{uv}$ is correct, and we can check if the candidate witness $Q_{uv} = w$ is correct by computing $S_{uw} + T_{wv}$. The algorithm clearly uses $O(\log n)$ matrix multiplications.

\paragraph{Finding witnesses in the general case.}
To find witnesses for all indices $(u,v)$, we reduce the general case to the case of unique witnesses. For simplicity, we only present a randomised version of this algorithm; for derandomisation see \citet{zwick2002all} and \citet{alon1996derandomization}.

Let $i \in \{ 0, 1, \dotsc, \lceil \log n \rceil - 1\}$. We use the following procedure to attempt to find witnesses for all $(u,v)$ that have exactly $r$ witnesses for $n / 2^{i + 1} \le r <  n / 2^{i}$:
\begin{enumerate}
    \item Let $m = \lceil c \log n \rceil$ for a sufficiently large constant $c$. For $j = 1,2, \dotsc, m$, construct a subset $V_{j} \subseteq V$ by picking $2^i$ values $v_1, v_2, \dotsc, v_{2^i}$ from $V$ with replacement, and let $V_j = \{ v_1, v_2, \dotsc, v_{2^i} \}$.
    \item For each $V_j$, use the unique witness detection for the product $S(\emptyind,V_j) \star T(V_j,\emptyind)$ to find candidate witnesses $Q_{uv}$ for all pairs $(u,v)$, and keep those $Q_{uv}$ that are witnesses for $S \star T$.
\end{enumerate}
Let $(u,v)$ be a pair with $r$ witnesses for $n / 2^{i + 1} \le r <  n / 2^{i}$. For each $j = 1, 2, \dotsc, m$, the probability that $V_j$ contains exactly one witness for $(u,v)$ is at least $(2e)^{-1}$ (see \citet{Seidel1995400}). Thus, the probability that we do not find a witness for $(u,v)$ is bounded by $(1-(2e)^{-1})^{\lceil c \log n \rceil} = n^{-\Omega(c)}$.

Repeating the above procedure for $i = 0, 1, \dotsc, \lceil \log n \rceil - 1$ ensures that the probability of not finding a witness for any fixed $(u,v)$ is at most $n^{-\Omega(c)}$. By the union bound, the probability that there is any pair of indices $(u,v)$ for which no witness is found is $n^{-\Omega(c)}$, i.e., with high probability the algorithm succeeds. Moreover, the total number of calls to the distance product is $O\bigl((\log n)^3\bigr)$, giving \lemmaref{lemma:witness}.



\section{Lower Bounds}\label{sec:lower-bounds}

\paragraph{Lower bounds for matrix multiplication implementations.}
While proving unconditional lower bounds for matrix multiplication in the congested clique model 
seems to be beyond the reach of current techniques, 
as discussed in \sectionref{sec:related-work}, 
it can be shown that the results given in \theoremref{thm:mm} are essentially optimal distributed implementations of the corresponding centralised algorithms. To be more formal, let $C$ be an arithmetic circuit for matrix multiplication; we say that an \emph{implementation} of $C$ in the congested clique model is a mapping of the gates of $C$ to the nodes of the congested clique. This naturally defines a congested clique algorithm for matrix multiplication, with the wires in $C$ between gates assigned to different nodes defining the communication cost of the algorithm.

Various authors, considering different parallel models, have shown that in any implementation of the trivial $\Theta(n^3)$ matrix multiplication on a parallel machine with $P$ processors there is at least one processor that has to send or receive $\Omega(n^2 / P^{2/3})$ matrix entries \cite{AggarwalCS90_PRAM,IronyTT04_3d_lower,tiskin1998}. As these models can simulate the congested clique, a similar lower bound holds for congested clique implementations of the trivial $O(n^3)$ matrix multiplication.
In the congested clique, each processor sends and receives $n$ messages per round 
(up to logarithmic factors) and $P = n$, yielding a lower bound of $\tilde\Omega(n^{1/3})$ rounds.

The trivial $\Theta(n^3)$ matrix multiplication is optimal for circuits using only 
semiring addition and multiplication.
The task of $n \times n$ matrix multiplication over the min-plus semiring can be reduced to APSP with a constant blowup~\cite[pp.202-205]{AhoHU74_algorithms}, 
hence the above bound applies also to any APSP algorithm that only uses minimum and addition operations.
This means that current techniques for similar problems, 
like the one used in the fast MST algorithm of Lotker et al.~\cite{lotker05} cannot be extended
to solve APSP.

\begin{corollary}
Any implementation of the trivial $\Theta(n^3)$ matrix multiplication,
and any APSP algorithm which only sums weights and takes the minimum of such sums,
require $\tilde\Omega(n^{1/3})$ communication rounds in the congested clique model.
\end{corollary}

However, known results on centralised APSP and distance product computation give reasons to suspect that this bound can be broken if we allow subtraction; 
in particular, translating the recent result of \citet{williams2014apsp} might allow for running time of order $n^{1/3}/2^{\Omega(\sqrt{\log n})}$ for APSP in the congested clique.

Concerning fast matrix multiplication algorithms, 
\citet{BallardDHS12_strassen_lower} have proven lower bounds 
for parallel implementations of \emph{Strassen-like} algorithms.
Their seminal work is based on building a DAG
representing the linear combinations of the inputs before the block multiplications,
and the linear combinations of the results of the multiplications (``decoding'') as the output matrix.
The parallel computation induces an assignment of the graph vertices
to the processes,
and the edges crossing the partition represent the communication.
Using an expansion argument,
Ballard et al. show that in any partition a graph representing an $\Omega(n^\mmeaux)$ algorithm
there is a process communicating $\Omega(n^{2-2/\mmeaux})$ values.
See also~\cite{BallardDHS14_summary} for a concise account of the technique.

The lower bound holds for Strassen's algorithm,
and for a family of similar algorithms,
but not for any matrix multiplication
algorithm (See~\cite[$\S.$ 5.1.1]{BallardDHS12_strassen_lower}).
A matrix multiplication algorithm is said to be \emph{Strassen-like}
if it is recursive, its decoding graph discussed above is connected,
and it computes no scalar multiplication twice.
As each process communicates at most $O(n)$ values in a round,
the implementation of an $\Omega(n^\mmeaux)$ strassen-like algorithm
must take $\Omega(n^{1-2/\mmeaux})$ rounds.

\begin{corollary}
Any implementation of a Strassen-like matrix multiplication algorithm using $\Omega(n^{\mmeaux})$ element multiplications requires $\tilde\Omega(n^{1 - 2/\mmeaux})$ communication rounds in the congested clique model.
\end{corollary}

\paragraph{Lower bound for broadcast congested clique.}
Recall that the broadcast congested clique is a version of the congested clique model with the additional constraint that all $n - 1$ messages sent by a node in a round must be identical.

\citet{FHW-12} have shown that approximating the diameter of an unweighted graph any better than factor $3/2$ requires $\tilde{\Omega}(n)$ rounds in the CONGEST model; the same can be applied to the broadcast congested clique. A variation of the approach was recently used by \citet{arXiv:1412.3445} to show that computing any approximation better than factor $2$ to all-pairs shortest paths in weighted graphs takes $\tilde\Omega(n)$ rounds as well. As discussed in Section~\ref{sec:routing}, $\tilde{o}(n)$-round matrix multiplication algorithms imply $\tilde{o}(n)$-round algorithms for exact unweighted and $(1+o(1))$-approximate weighted APSP. Together, this immediately implies that matrix multiplication on the broadcast congested clique is hard.

\begin{corollary}
In the broadcast congested clique model, matrix multiplication algorithms that are applicable to matrices over the Boolean semiring and APSP algorithms require $\tilde\Omega(n)$ communication rounds.
\end{corollary}
We remark that the phrase ``that is applicable to matrices over the Boolean semiring'' refers to the issue that, in principle, it is possible that matrix multiplication exponents may be different for different underlying semirings. However, at the very least the lower bound applies matrix multiplication over Booleans, integers, and rationals, as well as the min-plus semiring. We stress that, unlike the lower bounds presented beforehand, this bound holds without any assumptions on the algorithm itself.


\section{Conclusions}

In this work, we demonstrate that algebraic methods -- especially fast matrix multiplication -- can be used to design efficient algorithms in the congested clique model, resulting in algorithms that outperform the previous combinatorial algorithms; moreover, we have certainly not exhausted the known centralised literature of algorithms based on matrix multiplication, so similar techniques should also give improvements for other problems. It also remains open whether corresponding lower bounds exist; however, it increasingly looks like lower bounds for the congested clique would imply lower bounds for centralised algorithms, and are thus significantly more difficult to prove than for the CONGEST model.

While the present work focuses on a fully connected communication topology (clique), we expect that the same techniques can be applied more generally in the usual CONGEST model. For example, fast triangle detection in the CONGEST model is trivial in those areas of the network that are sparse. Only dense areas of the network are non-trivial, and in those areas we may have enough overall bandwidth for fast matrix multiplication algorithms. On the other hand, there are non-trivial lower bounds for distance computation problems in the CONGEST model~\cite{DHKNPPW-11}, though significant gaps still remain~\cite{nanongkai14}.

\section*{Acknowledgements}

Many thanks to Keijo Heljanko, Juho Hirvonen, Fabian Kuhn, Tuomo Lempi\"ainen, and Joel Rybicki for discussions.

\DeclareUrlCommand{\Doi}{\urlstyle{same}}
\renewcommand{\doi}[1]{\href{http://dx.doi.org/#1}{\footnotesize\sf doi:\Doi{#1}}}
\bibliographystyle{plainnat}
\small{
\renewcommand{\bibsep}{4pt}
\bibliography{triangles}

\begin{thebibliography}{77}
\providecommand{\natexlab}[1]{#1}
\providecommand{\url}[1]{\texttt{#1}}
\expandafter\ifx\csname urlstyle\endcsname\relax
  \providecommand{\doi}[1]{doi: #1}\else
  \providecommand{\doi}{doi: \begingroup \urlstyle{rm}\Url}\fi

\bibitem[Agarwal et~al.(1995)Agarwal, Balle, Gustavson, Joshi, and
  Palkar]{AgarwalBGJP95_3d}
Ramesh~C. Agarwal, Susanne~M. Balle, Fred~G. Gustavson, Mahesh~V. Joshi, and
  Prasad~V. Palkar.
\newblock A three-dimensional approach to parallel matrix multiplication.
\newblock \emph{{IBM} Journal of Research and Development}, 39\penalty0
  (5):\penalty0 575--582, 1995.
\newblock \doi{10.1147/rd.395.0575}.

\bibitem[Aggarwal et~al.(1990)Aggarwal, Chandra, and Snir]{AggarwalCS90_PRAM}
Alok Aggarwal, Ashok~K. Chandra, and Marc Snir.
\newblock Communication complexity of {PRAM}s.
\newblock \emph{Theoretical Computer Science}, 71\penalty0 (1):\penalty0 3--28,
  1990.
\newblock \doi{10.1016/0304-3975(90)90188-N}.

\bibitem[Aho et~al.(1974)Aho, Hopcroft, and Ullman]{AhoHU74_algorithms}
Alfred~V. Aho, John~E. Hopcroft, and Jeffrey~D. Ullman.
\newblock \emph{The Design and Analysis of Computer Algorithms}.
\newblock Addison-Wesley, 1974.
\newblock ISBN 0-201-00029-6.

\bibitem[Alon and Naor(1996)]{alon1996derandomization}
Noga Alon and Moni Naor.
\newblock Derandomization, witnesses for {Boolean} matrix multiplication and
  construction of perfect hash functions.
\newblock \emph{Algorithmica}, 16\penalty0 (4--5):\penalty0 434--449, 1996.
\newblock \doi{10.1007/BF01940874}.

\bibitem[Alon et~al.(1995)Alon, Yuster, and Zwick]{alon1995color}
Noga Alon, Raphael Yuster, and Uri Zwick.
\newblock Color-coding.
\newblock \emph{Journal of the {ACM}}, 42\penalty0 (4):\penalty0 844--856,
  1995.
\newblock \doi{10.1145/210332.210337}.

\bibitem[Alon et~al.(1997)Alon, Yuster, and Zwick]{alon1997finding}
Noga Alon, Raphael Yuster, and Uri Zwick.
\newblock Finding and counting given length cycles.
\newblock \emph{Algorithmica}, 17\penalty0 (3):\penalty0 209--223, 1997.
\newblock \doi{10.1007/BF02523189}.

\bibitem[Ballard et~al.(2012{\natexlab{a}})Ballard, Demmel, Holtz, Lipshitz,
  and Schwartz]{BallardDHLS12_Strassen_upper}
Grey Ballard, James Demmel, Olga Holtz, Benjamin Lipshitz, and Oded Schwartz.
\newblock Communication-optimal parallel algorithm for {S}trassen's matrix
  multiplication.
\newblock In \emph{Proc.\ 24th {ACM} Symposium on Parallelism in Algorithms and
  Architectures (SPAA 2012)}, pages 193--204, 2012{\natexlab{a}}.
\newblock \doi{10.1145/2312005.2312044}.

\bibitem[Ballard et~al.(2012{\natexlab{b}})Ballard, Demmel, Holtz, and
  Schwartz]{BallardDHS12_strassen_lower}
Grey Ballard, James Demmel, Olga Holtz, and Oded Schwartz.
\newblock Graph expansion and communication costs of fast matrix
  multiplication.
\newblock \emph{Journal of the {ACM}}, 59\penalty0 (6):\penalty0 32,
  2012{\natexlab{b}}.
\newblock \doi{10.1145/2395116.2395121}.

\bibitem[Ballard et~al.(2014)Ballard, Demmel, Holtz, and
  Schwartz]{BallardDHS14_summary}
Grey Ballard, James Demmel, Olga Holtz, and Oded Schwartz.
\newblock Communication costs of strassen's matrix multiplication.
\newblock \emph{Commun. {ACM}}, 57\penalty0 (2):\penalty0 107--114, 2014.
\newblock \doi{10.1145/2556647.2556660}.
\newblock URL \url{http://doi.acm.org/10.1145/2556647.2556660}.

\bibitem[Bj{\"o}rklund(2014)]{Bjorklund14}
Andreas Bj{\"o}rklund.
\newblock Determinant sums for undirected hamiltonicity.
\newblock \emph{SIAM Journal on Computing}, 43\penalty0 (1):\penalty0 280--299,
  2014.
\newblock \doi{10.1137/110839229}.

\bibitem[Bj{\"o}rklund and Husfeldt(2014)]{BjorklundH14}
Andreas Bj{\"o}rklund and Thore Husfeldt.
\newblock Shortest two disjoint paths in polynomial time.
\newblock In \emph{Proc.\ 41st International Colloquium on Automata, Languages,
  and Programming (ICALP 2014)}, volume 8572 of \emph{LNCS}, pages 211--222,
  2014.
\newblock \doi{10.1007/978-3-662-43948-7_18}.

\bibitem[Bj{\"o}rklund et~al.(2007)Bj{\"o}rklund, Husfeldt, Kaski, and
  Koivisto]{BjorklundHKK07}
Andreas Bj{\"o}rklund, Thore Husfeldt, Petteri Kaski, and Mikko Koivisto.
\newblock {F}ourier meets {M}{\"o}bius: fast subset convolution.
\newblock In \emph{Proc.\ 39th Annual ACM Symposium on Theory of Computing
  (STOC 2007)}, pages 67--74, 2007.
\newblock \doi{10.1145/1250790.1250801}.

\bibitem[Bj{\"o}rklund et~al.(2009)Bj{\"o}rklund, Husfeldt, and
  Koivisto]{BjorklundHK09}
Andreas Bj{\"o}rklund, Thore Husfeldt, and Mikko Koivisto.
\newblock Set partitioning via inclusion-exclusion.
\newblock \emph{SIAM Journal on Computing}, 39\penalty0 (2):\penalty0 546--563,
  2009.
\newblock \doi{10.1137/070683933}.

\bibitem[Bj{\"o}rklund et~al.(2013)Bj{\"o}rklund, Kaski, and
  Kowalik]{BjorklundKK13}
Andreas Bj{\"o}rklund, Petteri Kaski, and {\L}ukasz Kowalik.
\newblock Probably optimal graph motifs.
\newblock In \emph{Proc.\ 30th International Symposium on Theoretical Aspects
  of Computer Science (STACS 2013)}, volume~20 of \emph{LIPIcs}, pages 20--31,
  2013.
\newblock \doi{10.4230/LIPIcs.STACS.2013.20}.

\bibitem[Bj{\"o}rklund et~al.(2014)Bj{\"o}rklund, Kaski, and
  Kowalik]{BjorklundKK14}
Andreas Bj{\"o}rklund, Petteri Kaski, and {\L}ukasz Kowalik.
\newblock Counting thin subgraphs via packings faster than meet-in-the-middle
  time.
\newblock In \emph{Proc.\ 25th Annual ACM-SIAM Symposium on Discrete Algorithms
  (SODA 2014)}, pages 594--603, 2014.
\newblock \doi{10.1137/1.9781611973402.45}.

\bibitem[Bodlaender et~al.(2013)Bodlaender, Cygan, Kratsch, and
  Nederlof]{BodlaenderCKN13}
Hans~L. Bodlaender, Marek Cygan, Stefan Kratsch, and Jesper Nederlof.
\newblock Deterministic single exponential time algorithms for connectivity
  problems parameterized by treewidth.
\newblock In \emph{Proc.\ 40th International Colloquium on Automata, Languages,
  and Programming (ICALP 2013)}, pages 196--207, 2013.
\newblock \doi{10.1007/978-3-642-39206-1_17}.

\bibitem[B{\"u}rgisser et~al.(1997)B{\"u}rgisser, Clausen, and
  Shokrollahi]{burgisser1997algebraic}
Peter B{\"u}rgisser, Michael Clausen, and M.~Amin Shokrollahi.
\newblock \emph{Algebraic Complexity Theory}.
\newblock Springer, 1997.

\bibitem[Chan(2008)]{Chan08_apsp}
Timothy~M. Chan.
\newblock All-pairs shortest paths with real weights in ${O}(n^3/\log n)$
  time.
\newblock \emph{Algorithmica}, 50\penalty0 (2):\penalty0 236--243, 2008.
\newblock \doi{10.1007/s00453-007-9062-1}.

\bibitem[Chan(2010)]{Chan10_apsp}
Timothy~M. Chan.
\newblock More algorithms for all-pairs shortest paths in weighted graphs.
\newblock \emph{{SIAM} Journal on Computing}, 39\penalty0 (5):\penalty0
  2075--2089, 2010.
\newblock \doi{10.1137/08071990X}.

\bibitem[Cygan et~al.(2011)Cygan, Nederlof, Pilipczuk, Pilipczuk, van Rooij,
  and Wojtaszczyk]{CyganNPPRW11}
Marek Cygan, Jesper Nederlof, Marcin Pilipczuk, Michal Pilipczuk, Johan M.~M.
  van Rooij, and Jakub~Onufry Wojtaszczyk.
\newblock Solving connectivity problems parameterized by treewidth in single
  exponential time.
\newblock In \emph{Proc.\ 52nd Annual Symposium on Foundations of Computer
  Science (FOCS 2011)}, pages 150--159, 2011.
\newblock \doi{10.1109/FOCS.2011.23}.

\bibitem[Cygan et~al.(2013)Cygan, Kratsch, and Nederlof]{CyganKN13}
Marek Cygan, Stefan Kratsch, and Jesper Nederlof.
\newblock Fast {H}amiltonicity checking via bases of perfect matchings.
\newblock In \emph{Proc.\ 45th ACM Symposium on Theory of Computing (STOC
  2013)}, pages 301--310, 2013.
\newblock \doi{10.1145/2488608.2488646}.

\bibitem[Czumaj and Lingas(2007)]{CzumajL07}
Artur Czumaj and Andrzej Lingas.
\newblock Finding a heaviest triangle is not harder than matrix multiplication.
\newblock In \emph{Proc.\ 18th Annual ACM-SIAM Symposium on Discrete Algorithms
  (SODA 2007)}, pages 986--994, 2007.

\bibitem[{Das Sarma} et~al.(2011){Das Sarma}, Holzer, Kor, Korman, Nanongkai,
  Pandurangan, Peleg, and Wattenhofer]{DHKNPPW-11}
Atish {Das Sarma}, Stephan Holzer, Liah Kor, Amos Korman, Danupon Nanongkai,
  Gopal Pandurangan, David Peleg, and Roger Wattenhofer.
\newblock Distributed verification and hardness of distributed approximation.
\newblock In \emph{Proc.\ 43th ACM Symposium on Theory of Computing (STOC
  2011)}, pages 363--372, 2011.
\newblock \doi{10.1145/1993636.1993686}.

\bibitem[Dolev et~al.(2012)Dolev, Lenzen, and Peled]{tritri}
Danny Dolev, Christoph Lenzen, and Shir Peled.
\newblock ``{T}ri, tri again'': Finding triangles and small subgraphs in a
  distributed setting.
\newblock In \emph{Proc.\ 26th International Symposium on Distributed Computing
  (DISC 2012)}, pages 195--209, 2012.
\newblock \doi{10.1007/978-3-642-33651-5_14}.

\bibitem[Drucker et~al.(2014)Drucker, Kuhn, and Oshman]{drucker13}
Andrew Drucker, Fabian Kuhn, and Rotem Oshman.
\newblock On the power of the congested clique model.
\newblock In \emph{Proc.\ 33rd ACM Symposium on Principles of Distributed
  Computing (PODC 2014)}, pages 367--376, 2014.
\newblock \doi{10.1145/2611462.2611493}.

\bibitem[Eisenbrand and Grandoni(2004)]{EisenbrandG04}
Friedrich Eisenbrand and Fabrizio Grandoni.
\newblock On the complexity of fixed parameter clique and dominating set.
\newblock \emph{Theoretical Computer Science}, 326\penalty0 (1--3):\penalty0
  57--67, 2004.
\newblock \doi{10.1016/j.tcs.2004.05.009}.

\bibitem[Fischer and Meyer(1971)]{fm1971boolean}
M.~J. Fischer and A.~R. Meyer.
\newblock Boolean matrix multiplication and transitive closure.
\newblock In \emph{Proc.\ 12th Symposium on Switching and Automata Theory (FOCS
  1971)}, pages 129--131, 1971.
\newblock \doi{10.1109/SWAT.1971.4}.

\bibitem[Fomin et~al.(2012)Fomin, Lokshtanov, Raman, Saurabh, and
  Rao]{FominLRSR12}
Fedor~V. Fomin, Daniel Lokshtanov, Venkatesh Raman, Saket Saurabh, and
  B.~V.~Raghavendra Rao.
\newblock Faster algorithms for finding and counting subgraphs.
\newblock \emph{Journal of Computer and System Sciences}, 78\penalty0
  (3):\penalty0 698--706, 2012.
\newblock \doi{10.1016/j.jcss.2011.10.001}.

\bibitem[Fomin et~al.(2014)Fomin, Lokshtanov, and Saurabh]{FominLS14}
Fedor~V. Fomin, Daniel Lokshtanov, and Saket Saurabh.
\newblock Efficient computation of representative sets with applications in
  parameterized and exact algorithms.
\newblock In \emph{Proc.\ 25th Annual ACM-SIAM Symposium on Discrete Algorithms
  (SODA 2014)}, pages 142--151, 2014.
\newblock \doi{10.1137/1.9781611973402.10}.

\bibitem[Fredman(1976)]{Fredman76_apsp}
Michael~L. Fredman.
\newblock New bounds on the complexity of the shortest path problem.
\newblock \emph{{SIAM} Journal on Computing}, 5\penalty0 (1):\penalty0 83--89,
  1976.
\newblock \doi{10.1137/0205006}.

\bibitem[Frischknecht et~al.(2012)Frischknecht, Holzer, and
  Wattenhofer]{FHW-12}
Silvio Frischknecht, Stephan Holzer, and Roger Wattenhofer.
\newblock Networks cannot compute their diameter in sublinear time.
\newblock In \emph{Proc.\ 23rd Annual ACM-SIAM Symposium on Discrete Algorithms
  (SODA 2012)}, pages 1150--1162, 2012.

\bibitem[Furman(1970)]{furman1970application}
M.~E. Furman.
\newblock Application of a method of fast multiplication of matrices to problem
  of finding graph transitive closure.
\newblock \emph{Doklady Akademii Nauk SSSR}, 194\penalty0 (3):\penalty0 524,
  1970.

\bibitem[Gall(2014)]{legall2014powers}
Fran{\c c}ois~Le Gall.
\newblock Powers of tensors and fast matrix multiplication.
\newblock In \emph{Proc.\ 39th Symposium on Symbolic and Algebraic Computation
  (ISSAC 2014)}, pages 296--303, 2014.
\newblock \doi{10.1145/2608628.2608664}.

\bibitem[Han(2008)]{Han08_apsp}
Yijie Han.
\newblock An ${O}(n^3(\log \log n/\log n)^{5/4})$ time algorithm for all
  pairs shortest path.
\newblock \emph{Algorithmica}, 51\penalty0 (4):\penalty0 428--434, 2008.
\newblock \doi{10.1007/s00453-007-9063-0}.

\bibitem[Han and Takaoka(2012)]{HanT12_apsp}
Yijie Han and Tadao Takaoka.
\newblock An ${O}(n^3 \log \log n / \log^2 n)$ time algorithm for all pairs
  shortest paths.
\newblock In \emph{Proc.\ 13th Scandinavian Symposium on Algorithm Theory (SWAT
  2012)}, pages 131--141, 2012.
\newblock \doi{10.1007/978-3-642-31155-0_12}.

\bibitem[Hegeman and Pemmaraju(2014)]{HegemanP14}
James~W. Hegeman and Sriram~V. Pemmaraju.
\newblock Lessons from the congested clique applied to mapreduce.
\newblock In \emph{Proc.\ 21st Colloquium on Structural Information and
  Communication Complexity (SIROCCO 2014)}, pages 149--164, 2014.
\newblock \doi{10.1007/978-3-319-09620-9_13}.

\bibitem[Hegeman et~al.(2014)Hegeman, Pemmaraju, and Sardeshmukh]{hegeman14}
James~W. Hegeman, Sriram~V. Pemmaraju, and Vivek~B. Sardeshmukh.
\newblock Near-constant-time distributed algorithms on a congested clique.
\newblock In \emph{Proc.\ 28th International Symposium on Distributed Computing
  (DISC 2014)}, pages 514--530, 2014.
\newblock \doi{10.1007/978-3-662-45174-8_35}.

\bibitem[{Holzer} and {Pinsker}(2014)]{arXiv:1412.3445}
Stephan {Holzer} and N.~{Pinsker}.
\newblock Approximation of distances and shortest paths in the broadcast
  congest clique, 2014.
\newblock arXiv:1412.3445 [cs.DC].

\bibitem[Holzer and Wattenhofer(2012)]{HolzerW12}
Stephan Holzer and Roger Wattenhofer.
\newblock Optimal distributed all pairs shortest paths and applications.
\newblock In \emph{Proc.\ 31st ACM Symposium on Principles of Distributed
  Computing (PODC 2012)}, pages 355--364, 2012.
\newblock \doi{10.1145/2332432.2332504}.

\bibitem[Holzer et~al.(2014)Holzer, Peleg, Roditty, and Wattenhofer]{holzer14}
Stephan Holzer, David Peleg, Liam Roditty, and Roger Wattenhofer.
\newblock Brief announcement: Distributed $3/2$-approximation of the diameter.
\newblock In \emph{Proc.\ 28th International Symposium on Distributed Computing
  (DISC 2014)}, pages 562--564, 2014.

\bibitem[Irony et~al.(2004)Irony, Toledo, and Tiskin]{IronyTT04_3d_lower}
Dror Irony, Sivan Toledo, and Alexandre Tiskin.
\newblock Communication lower bounds for distributed-memory matrix
  multiplication.
\newblock \emph{Journal of Parallel and Distributed Computing}, 64\penalty0
  (9):\penalty0 1017--1026, 2004.
\newblock \doi{10.1016/j.jpdc.2004.03.021}.

\bibitem[Itai and Rodeh(1978)]{itai1978finding}
Alon Itai and Michael Rodeh.
\newblock Finding a minimum circuit in a graph.
\newblock \emph{SIAM Journal on Computing}, 7\penalty0 (4):\penalty0 413--423,
  1978.
\newblock \doi{10.1137/0207033}.

\bibitem[Koutis(2008)]{Koutis08}
Ioannis Koutis.
\newblock Faster algebraic algorithms for path and packing problems.
\newblock In \emph{Proc.\ 35th International Colloquium on Automata, Languages
  and Programming (ICALP 2008)}, volume 5125 of \emph{LNCS}, pages 575--586,
  2008.
\newblock \doi{10.1007/978-3-540-70575-8_47}.

\bibitem[Kowaluk et~al.(2011)Kowaluk, Lingas, and Lundell]{KowalukLL11}
Miroslaw Kowaluk, Andrzej Lingas, and Eva-Marta Lundell.
\newblock Counting and detecting small subgraphs via equations and matrix
  multiplication.
\newblock In \emph{Proc.\ 22nd Annual ACM-SIAM Symposium on Discrete Algorithms
  (SODA 2011)}, pages 1468--1476, 2011.
\newblock \doi{10.1137/1.9781611973082.114}.

\bibitem[Kutten and Peleg(1998)]{KP98}
Shay Kutten and David Peleg.
\newblock Fast distributed construction of small {$k$}-dominating sets and
  applications.
\newblock \emph{Journal of Algorithms}, 28\penalty0 (1):\penalty0 40--66, 1998.
\newblock \doi{10.1006/jagm.1998.0929}.

\bibitem[Lenzen(2013)]{lenzen2013optimal}
Christoph Lenzen.
\newblock Optimal deterministic routing and sorting on the congested clique.
\newblock In \emph{Proc.\ 32nd ACM Symposium on Principles of Distributed
  Computing (PODC 2013)}, pages 42--50, 2013.
\newblock \doi{10.1145/2484239.2501983}.

\bibitem[Lenzen and Patt{-}Shamir(2013)]{LenzenP13_routing_tables}
Christoph Lenzen and Boaz Patt{-}Shamir.
\newblock Fast routing table construction using small messages.
\newblock In \emph{Proc.\ 45rd {ACM} Symposium on Theory of Computing (STOC
  2013)}, pages 381--390, 2013.
\newblock \doi{10.1145/2488608.2488656}.

\bibitem[Lenzen and Peleg(2013)]{LP13:podc}
Christoph Lenzen and David Peleg.
\newblock Efficient distributed source detection with limited bandwidth.
\newblock In \emph{Proc.\ 32nd ACM Symposium on Principles of Distributed
  Computing (PODC 2013)}, pages 375--382, 2013.
\newblock \doi{10.1145/2484239.2484262}.

\bibitem[Lenzen and Wattenhofer(2011)]{LenzenW11}
Christoph Lenzen and Roger Wattenhofer.
\newblock Tight bounds for parallel randomized load balancing.
\newblock In \emph{Proc.\ 43rd {ACM} Symposium on Theory of Computing (STOC
  2011)}, pages 11--20, 2011.
\newblock \doi{10.1145/1993636.1993639}.

\bibitem[Lokshtanov and Nederlof(2010)]{LokshtanovN10}
Daniel Lokshtanov and Jesper Nederlof.
\newblock Saving space by algebraization.
\newblock In \emph{Proc.\ 42nd ACM Symposium on Theory of Computing (STOC
  2010)}, 2010.
\newblock \doi{10.1145/1806689.1806735}.

\bibitem[Lotker et~al.(2005)Lotker, Patt-Shamir, Pavlov, and Peleg]{lotker05}
Zvi Lotker, Boaz Patt-Shamir, Elan Pavlov, and David Peleg.
\newblock Minimum-weight spanning tree construction in {$O(\log\log n)$}
  communication rounds.
\newblock \emph{SIAM Journal on Computing}, 35\penalty0 (1):\penalty0 120--131,
  2005.
\newblock \doi{10.1137/S0097539704441848}.

\bibitem[Luo and Drake(1995)]{LuoD95_layout_parallel}
Qingshan Luo and John~B. Drake.
\newblock A scalable parallel {S}trassen's matrix multiplication algorithm for
  distributed-memory computers.
\newblock In \emph{Symposium on Applied Computing (SAC)}, pages 221--226, 1995.
\newblock \doi{10.1145/315891.315965}.

\bibitem[Matou{\v{s}}ek(2002)]{Matousek02_geometry}
Jir{\'{\i}} Matou{\v{s}}ek.
\newblock \emph{Lectures on Discrete Geometry}.
\newblock Graduate Texts in Mathematics. Springer, 2002.
\newblock ISBN 9780387953731.

\bibitem[McColl(1995)]{mccoll1995scalable}
William~F. McColl.
\newblock Scalable computing.
\newblock In \emph{Computer Science Today}, volume 1000 of \emph{LNCS}, pages
  46--61. Springer, 1995.
\newblock \doi{10.1007/BFb0015236}.

\bibitem[McColl(1996)]{mccoll1996}
William~F. McColl.
\newblock A {BSP} realisation of {S}trassen's algorithm.
\newblock In \emph{Abstract Machine Models for Parallel and Distributed
  Computing}, volume~48 of \emph{Concurrent Systems Engineering}, pages 43--46.
  IOS Press, 1996.

\bibitem[Munro(1971)]{Munro197156}
Ian Munro.
\newblock Efficient determination of the transitive closure of a directed
  graph.
\newblock \emph{Information Processing Letters}, 1\penalty0 (2):\penalty0
  56--58, 1971.
\newblock \doi{10.1016/0020-0190(71)90006-8}.

\bibitem[Nanongkai(2014)]{nanongkai14}
Danupon Nanongkai.
\newblock Distributed approximation algorithms for weighted shortest paths.
\newblock In \emph{Proc.\ 46th ACM Symposium on Theory of Computing (STOC
  2014)}, pages 565--573, 2014.

\bibitem[Ne{\v s}et{\v r}il and Poljak(1985)]{nevsetvril1985complexity}
Jaroslav Ne{\v s}et{\v r}il and Svatopluk Poljak.
\newblock On the complexity of the subgraph problem.
\newblock \emph{Commentationes Mathematicae Universitatis Carolinae},
  26\penalty0 (2):\penalty0 415--419, 1985.

\bibitem[Patt-Shamir and Teplitsky(2011)]{patt-shamir11}
Boaz Patt-Shamir and Marat Teplitsky.
\newblock The round complexity of distributed sorting.
\newblock In \emph{Proc.\ 30th ACM Symposium on Principles of Distributed
  Computing (PODC 2011)}, pages 249--256, 2011.
\newblock \doi{10.1145/1993806.1993851}.

\bibitem[Peleg(2000)]{peleg00}
David Peleg.
\newblock \emph{Distributed Computing: A Locality-Sensitive Approach}.
\newblock Society for Industrial and Applied Mathematics, 2000.

\bibitem[Peleg and Rubinovich(2000)]{PelegR-00}
David Peleg and Vitaly Rubinovich.
\newblock Near-tight lower bound on the time complexity of distributed {MST}
  construction.
\newblock \emph{SIAM Journal on Computing}, 30\penalty0 (5):\penalty0
  1427--1442, 2000.
\newblock \doi{10.1137/S0097539700369740}.

\bibitem[Peleg et~al.(2012)Peleg, Roditty, and Tal]{PelegRT12}
David Peleg, Liam Roditty, and Elad Tal.
\newblock Distributed algorithms for network diameter and girth.
\newblock In \emph{Proc.\ 39th International Colloquium on Automata, Languages
  and Programming (ICALP 2012)}, pages 660--672, 2012.
\newblock \doi{10.1007/978-3-642-31585-5_58}.

\bibitem[Pemmaraju and Sardeshmukh(2014)]{PemmarajuS14_MST_logloglogn}
Sriram~V. Pemmaraju and Vivek~B. Sardeshmukh.
\newblock Minimum-weight spanning tree construction in
  {\textdollar}o({\textbackslash}log {\textbackslash}log {\textbackslash}log
  n){\textdollar} rounds on the congested clique.
\newblock \emph{CoRR}, abs/1412.2333, 2014.
\newblock URL \url{http://arxiv.org/abs/1412.2333}.

\bibitem[Raghavan and Thompson(1985)]{raghavan85}
P~Raghavan and C~D Thompson.
\newblock {Provably Good Routing in Graphs: Regular Arrays}.
\newblock In \emph{{Proc.\ 7th ACM Symposium on Theory of Computing (STOC
  1985)}}, pages 79--87, 1985.

\bibitem[Seidel(1995)]{Seidel1995400}
Raimund Seidel.
\newblock On the all-pairs-shortest-path problem in unweighted undirected
  graphs.
\newblock \emph{Journal of Computer and System Sciences}, 51\penalty0
  (3):\penalty0 400--403, 1995.
\newblock \doi{10.1006/jcss.1995.1078}.

\bibitem[Strassen(1969)]{strassen}
Volker Strassen.
\newblock Gaussian elimination is not optimal.
\newblock \emph{Numerische Mathematik}, 13\penalty0 (4):\penalty0 354--356,
  1969.
\newblock \doi{10.1007/BF02165411}.

\bibitem[Takaoka(2004)]{Takaoka04_apsp}
Tadao Takaoka.
\newblock A faster algorithm for the all-pairs shortest path problem and its
  application.
\newblock In \emph{Proc.\ 10th International Conference on Computing and
  Combinatorics (COCOON 2004)}, pages 278--289, 2004.
\newblock \doi{10.1007/978-3-540-27798-9_31}.

\bibitem[Takaoka(2005)]{Takaoka05_apsp}
Tadao Takaoka.
\newblock An ${O}(n^3 \log \log n/\log n)$ time algorithm for the all-pairs
  shortest path problem.
\newblock \emph{Information Processing Letters}, 96\penalty0 (5):\penalty0
  155--161, 2005.
\newblock \doi{10.1016/j.ipl.2005.08.008}.

\bibitem[Tiskin(1998)]{tiskin1998}
Alexandre Tiskin.
\newblock Bulk-synchronous parallel multiplication of boolean matrices.
\newblock In \emph{Proc.\ 25th Colloquium on Automata, Languages and
  Programming (ICALP 1998)}, pages 494--506. Springer Berlin Heidelberg, 1998.
\newblock \doi{10.1007/BFb0055078}.

\bibitem[Tiskin(1999)]{tiskin-phd}
Alexandre Tiskin.
\newblock \emph{The Design and Analysis of Bulk-Synchronous Parallel
  Algorithms}.
\newblock PhD thesis, University of Oxford, 1999.

\bibitem[Vassilevska~Williams and Williams(2013)]{vassilevska2013finding}
Virginia Vassilevska~Williams and Ryan Williams.
\newblock Finding, minimizing, and counting weighted subgraphs.
\newblock \emph{SIAM Journal of Computing}, 42\penalty0 (3):\penalty0 831--854,
  2013.
\newblock \doi{10.1137/09076619X}.

\bibitem[Williams(2009)]{Williams09}
Ryan Williams.
\newblock Finding paths of length {$k$} in {$O^*(2^k)$} time.
\newblock \emph{Information Processing Letters}, 109\penalty0 (6):\penalty0
  315--318, 2009.
\newblock \doi{10.1016/j.ipl.2008.11.004}.

\bibitem[Williams(2014)]{williams2014apsp}
Ryan Williams.
\newblock Faster all-pairs shortest paths via circuit complexity.
\newblock In \emph{Proc.\ 46th {ACM} Symposium on Theory of Computing (STOC
  2014)}, pages 664--673. ACM, 2014.
\newblock \doi{10.1145/2591796.2591811}.

\bibitem[Yuval(1976)]{Yuval76_apsp}
Gideon Yuval.
\newblock An algorithm for finding all shortest paths using $n^{2.81}$
  infinite-precision multiplications.
\newblock \emph{Information Processing Letters}, 4\penalty0 (6):\penalty0
  155--156, 1976.
\newblock \doi{10.1016/0020-0190(76)90085-5}.

\bibitem[Zwick(2001)]{Zwick01_graph_distances_survey}
Uri Zwick.
\newblock Exact and approximate distances in graphs -- {A} survey.
\newblock In \emph{Proc.\ 9th European Symposium on Algorithms (ESA 2001)},
  pages 33--48, 2001.
\newblock \doi{10.1007/3-540-44676-1_3}.

\bibitem[Zwick(2002)]{zwick2002all}
Uri Zwick.
\newblock All pairs shortest paths using bridging sets and rectangular matrix
  multiplication.
\newblock \emph{Journal of the {ACM}}, 49\penalty0 (3):\penalty0 289--317,
  2002.
\newblock \doi{10.1145/567112.567114}.

\bibitem[Zwick(2006)]{Zwick06_APSP}
Uri Zwick.
\newblock A slightly improved sub-cubic algorithm for the all pairs shortest
  paths problem with real edge lengths.
\newblock \emph{Algorithmica}, 46\penalty0 (2):\penalty0 181--192, 2006.
\newblock \doi{10.1007/s00453-005-1199-1}.

\end{thebibliography}
}

\end{document}